\documentclass[10pt,twocolumn,twoside] {IEEEtran}
\usepackage{cite}
\usepackage{graphicx}
\usepackage{psfrag}
\usepackage{url}
\usepackage{amsmath}
\usepackage{array}
\usepackage{amssymb}
\usepackage{amsfonts}
\usepackage{epstopdf}

\newtheorem{remark}{Remark}
\newtheorem{definition}{Definition}
\newtheorem{theorem}{Theorem}
\newtheorem{example}{Example}

\usepackage{algorithm}
\usepackage{algpseudocode}
\usepackage{subfigure}

\title{On Precoding for Constant K-User MIMO Gaussian Interference Channel with Finite Constellation Inputs}

\begin{document}

\author{
\authorblockN{Abhinav Ganesan and B. Sundar Rajan\\}
\IEEEauthorblockA{\small{Email: \{abhig\_88, bsrajan\}@ece.iisc.ernet.in}}
\thanks{\hrule \vspace{0.2cm} Part of the content of this paper appeared in the Proceedings of IEEE ICC 2013 held during June $9$-$13$, $2013$ at Budapest, Hungary.}
}

\maketitle
\thispagestyle{empty}	

\begin{abstract}
This paper considers linear precoding for constant channel-coefficient $K$-User MIMO Gaussian Interference Channel (MIMO GIC) where each transmitter-$i$ (Tx-$i$), requires to send $d_i$ independent complex symbols per channel use that take values from fixed finite constellations with uniform distribution, to receiver-$i$ (Rx-$i$) for $i=1,2,\cdots ,K$. We define the maximum rate achieved by Tx-$i$ using any linear precoder, when the interference channel-coefficients are zero, as the signal to noise ratio (SNR) tends to infinity to be the Constellation Constrained Saturation Capacity (CCSC) for Tx-$i$. We derive a high SNR approximation for the rate achieved by Tx-$i$ when interference is treated as noise and this rate is given by the mutual information between Tx-$i$ and Rx-$i$, denoted as $I[X_i;Y_i]$. A set of necessary and sufficient conditions on the precoders under which $I[X_i;Y_i]$ tends to CCSC for Tx-$i$ is derived. Interestingly, the precoders designed for interference alignment (IA) satisfy these necessary and sufficient conditions.
Further, we propose gradient-ascent based algorithms to optimize the sum-rate achieved by precoding with finite constellation inputs and treating interference as noise. Simulation study using the proposed algorithms for a $3$-user MIMO GIC with two antennas at each node with $d_i=1$ for all $i$, and with BPSK and QPSK inputs, show more than $0.1$ bits/sec/Hz gain in the ergodic sum-rate over that yielded by precoders obtained from some known IA algorithms, at moderate SNRs. 
\end{abstract}

\section{Introduction} \label{sec1}
Interference Alignment (IA) has been a focus of intense research on Gaussian interference channels (GICs) in the recent years on the account of the capacity of interference channels being unknown in general and the potential of IA to get close to the sum-capacity for a broad class of interference channels as the signal to noise ratio ($SNR$) tends to infinity. The scaling of the sum-capacity with $log ~SNR$ is known as the sum degrees of freedom (DoF) of the GIC \cite{CaJ1}. The sum-DoF is known to be $K/2$, with probability $1$, for the K-user GIC \cite{CaJ1} when all the nodes are equipped with a single antenna and time-varying channel gains are assumed. The result was proved by using linear precoding at the transmitters over an arbitrarily large number of symbol extensions and zero forcing at the receivers. In a later work \cite{MGMK}, the sum-DoF is shown to be $K/2$ with probability $1$ even in the case of constant channel-coefficients (that are drawn from a continuous distribution), i.e., the channel 
gains do not vary with time, using a non-linear IA technique. When all the nodes are equipped with $M$ antennas, the sum-DoF was shown to be $3M/2$ for the $3$-user GIC \cite{CaJ1}. This result was proved by linear precoding over the transmit antennas without the use of symbol extensions and holds true even in the case of constant channel coefficients. Later \cite{GMK}, with the assumption of constant channel coefficients and using a non-linear IA technique, the sum-DoF when $K\geq \frac{M+N}{\gcd(M,N)}$ was found to be equal to $\frac{MN}{M+N} K$ where, $\gcd(M,N)$ denotes the greatest common divisor of $M$ and $N$, with $M$ being the number of antennas at each transmitter and $N$ being the number of receive antennas at each receiver. All the works cited above assumed full channel state information at all the transmitters (CSIT) and receivers (CSIR). The notion of sum-DoF involves scaling of sum-rate as $log ~SNR$ at high SNR and therefore, Gaussian input alphabets or lattice codes are always used in the 
study of sum-DoF. However, in all practical scenarios finite constellations like $M$-QAM and $M$-PSK are used at the inputs. {\em With the constraint of finite constellation inputs, it is not known whether IA is optimal in some sense}. 

Linear precoding for optimizing the mutual information between the input and the output has been studied for the single user MIMO channel with finite constellation inputs in \cite{PaV}-\cite{XZD}. Constellation rotation for optimizing the sum-capacity for SISO Multiple Access Channel (MAC) with finite constellation inputs has been examined in \cite{HaR} and  linear precoding for weighted sum-rate maximization in MIMO MAC with finite constellation inputs has been studied in \cite{WZX}. Note that linear precoding for the SISO MAC corresponds to constellation rotation at the transmitter. 

Recently, there has been some progress on the analysis of finite constellation effects in $2$-user SISO GIC \cite{KnS}, \cite{AbR}. In \cite{KnS}, constellation rotation was found to increase the constellation constrained sum-capacity of $2$-user SISO Gaussian strong interference channel \cite{HS,CoE}, and in \cite{AbR}, a metric to find the optimum angle of rotation was proposed. In this paper, we examine achievable rate-tuples with linear precoding for $K$-user MIMO Gaussian Interference Channel (GIC) with finite constellation inputs. Specifically, we treat interference as noise, i.e., each transmitter reveals its codebook only to its intended receiver. The maximum rate achievable under such a circumstance for transmitter-$i$ (Tx-$i$) is given by mutual information between the input generated by Tx-$i$ and the output at receiver Rx-$i$. The channel conditions and values of SNR under which the decoding scheme of treating interference as noise with Gaussian alphabet inputs is sum-capacity optimal was found 
for the $2$-user  SISO GIC in \cite{SKC1}-\cite{AnV1}, for the $K$-user SISO GIC in \cite{SKC2} and for the $2$-user MIMO GIC in \cite{SKCV},\cite{AnV2}. In a $K$-user SISO GIC, for given values of channel gains with Gaussian input alphabets, as the SNR tends to infinity, treating interference as noise is not sum-capacity optimal \cite{GaY}. {\em With the constraint of finite constellation inputs, it is not clear whether treating interference as noise is optimal in some sense.} 

First, we need to define a notion of optimality under the constraint of fixed finite constellation inputs and then analyse decoding and transmit schemes with that notion of optimality. Consider a scenario where each transmitter-$i$ (Tx-$i$), requires to send $d_i$ independent complex symbols per channel use that take values from fixed finite constellations with uniform distribution, for $i=1,2,\cdots ,K$, to receiver-$i$ (Rx-$i$). Throughout this paper, we assume that none of the direct channel gains are zero. For a $K$-user MIMO GIC with finite constellation inputs, as a measure of optimality of linear precoding in the high SNR sense, we introduce the notion of  {\em Constellation Constrained Saturation Capacity} (CCSC) which is defined as follows.
\begin{definition}
The maximum rate achieved by Tx-$i$ as SNR tends to infinity, using any linear precoder, when the interference channel-coefficients are zero is termed as the Constellation Constrained Saturation Capacity (CCSC) for Tx-$i$.
\end{definition}
For the ease of exposition, throughout the paper, we assume that the constellations used for the symbols are all the same at all the transmitters, and is of cardinality $M$. Hence, the CCSC for Tx-$i$ is given by $log_2M^{ d_i}$. 

In this paper, with the assumption of constant $K$-user MIMO GIC with full global knowledge of channels gains, and finite constellation inputs, we derive a set of necessary and sufficient conditions on the precoders under which treating interference as noise at Rx-$i$ will achieve a rate for Tx-$i$ that tends to CCSC for Tx-$i$, for all $i$, as SNR tends to infinity. Precoders satisfying these necessary and sufficient conditions exist for all direct and cross channel gains, and are termed as {\em CCSC optimal precoders}. Hence, in the case of finite constellation inputs with the use of appropriate precoders, the rate tuples obtained by treating interference as noise tend to values that are independent of the channel gains. For a $K$-user SISO GIC, this result is in contrast with the Gaussian input alphabet case where the rate tuples obtained by treating interference as noise tend to values dictated by the channel gains, as the SNR tends to infinity \footnote{These results appeared first in an older \textit{
arxiv} version of this paper \cite{AbR2}. Similar results also appeared recently in \cite{WXGMD}.}. Interestingly, the precoders that achieve IA, if feasible, are also CCSC optimal precoders. However, finding precoders that align interference is known to be NP-hard \cite{RSL} in general whereas, the precoders that satisfy the derived necessary and sufficient conditions are easy to find for any given channel-coefficients.

Since finite SNR is of more practical interest, we propose gradient-ascent based algorithms to optimize the precoders for the sum-rate achieved by treating interference as noise with finite constellation inputs. We point out that optimization of mutual information between the input and output in single user MIMO channels with finite constellation inputs was recently solved in \cite{XZD}. Prior to this work, gradient-ascent based algorithms were proposed for optimization of mutual information between the input and output in single user MIMO channels with finite constellation input in \cite{PaV,PRS}. Further, optimizing the minimum Euclidean distance metric using linear precoding in single user MIMO channels was pursued in \cite{CBRB}-\cite{SrR}. The connection between optimizing the minimum Euclidean distance metric using linear precoding and optimizing the mutual information using linear precoding  with finite constellation input was revealed in \cite{PRS}. It was shown that the precoding matrix that 
maximizes the mutual information with finite constellation input converges to the matrix that maximizes the minimum distance between the received constellation vectors, at large SNR. We generalize this connection in the context of $K$-user MIMO GIC (Theorem \ref{thm3}, Section \ref{subsec4a}).

The main contributions of the paper are summarized below.

\begin{itemize}
\item For a constant $K$-user MIMO GIC using finite constellation inputs with precoding, a high SNR approximation for the rate tuples achieved by treating interference as noise at the receivers is derived (Theorem \ref{thm1} in Section \ref{sec3}). Based on this approximation, we derive a set of necessary and sufficient conditions under which the precoders are CCSC optimal (Theorem \ref{thm2}, Section \ref{sec3}). These conditions are satisfied with probability $1$ when the entries of the precoders are chosen from any continuous distribution. It is observed that the precoders that achieve IA, if feasible, are CCSC optimal.
\item For the finite SNR case, we propose a gradient-ascent based algorithms to improve the sum-rate achieved by treating interference as noise using finite constellation inputs with precoding. Simulation study using the proposed algorithms for a $3$-user MIMO GIC with two antennas at each node with $d_i=1$ for all $i$, and with BPSK and QPSK inputs, shows an improvement of over $0.1$ bits/sec/Hz in the ergodic sum-rate over that obtained using some known IA algorithms, at moderate SNRs (Section \ref{subsec4b}).
\end{itemize}

The paper is organized as follows. The system model is formally introduced in Section \ref{sec2}. In Section \ref{sec3}, a set of necessary and sufficient conditions for CCSC optimal precoders is derived. In Section \ref{sec4}, our gradient-ascent based algorithms for optimizing the sum-rate using precoders and treating interference as noise at the receivers is given, and simulation results comparing their performance with respect to that of precoders obtained using some known IA algorithms is presented. Section \ref{sec5} concludes the paper.

 {\em Notations:}  For a random vector $X$ which takes value from the set $\cal X$, we assume some ordering of its elements and use $x^i$ to represent the $i$-th element of $\cal X$. Realization of the random vector $X$ is denoted as $x$. The notation $diag(V_1,V_2,\cdots,V_n)$ denotes a block diagonal matrix formed by the matrices $V_i$, $i=1,2,\cdots,n$. The $i^{\text{th}}$ coordinate of a complex vector $X$ is denoted by $X(i)$. The $2$-norm of a complex vector $X$ is denoted by $||X||$. The cardinality of a set ${\cal X}$ is denoted by $|\cal X|$. For a complex number $a$, $\Re\{a\}$ and $\Im\{a\}$ denote the real and imaginary parts of $a$ respectively. For two complex numbers $a$ and $b$, the notation $a>b$ denotes that $|\Re\{a\}|>|\Re\{b\}|$ and $|\Im\{a\}|>|\Im\{b\}|$. The notation $\underline{0}$ represents the zero vector whose size will be clear from the context. All the logarithms in the paper are to the base $2$.

\section{System Model} \label{sec2}
 \begin{figure}[htbp]
\centering
\includegraphics[totalheight=3.1in,width=2.8in]{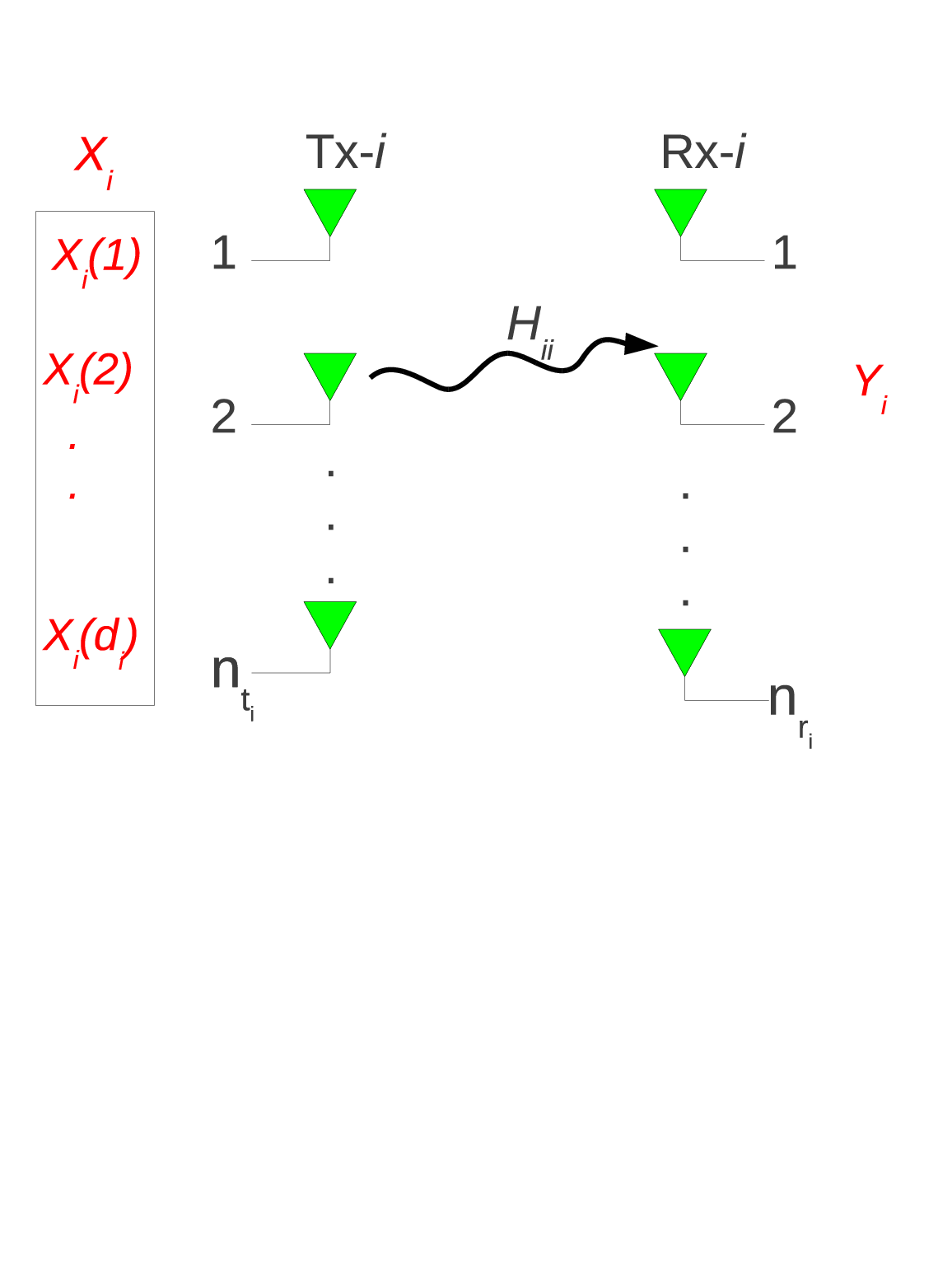}
\vspace{-3cm}
\caption{System Model.}
\label{fig_sys_model}
\end{figure}
Tx-$i$ intends to communicate with Rx-$i$, for $i=1,2,\cdots, K$, as shown in Fig. \ref{fig_sys_model}. Without loss of generality, let $P$ denote the power constraint at all the transmitters. The signal received at Rx-$j$ is given by
\begin{align*}
 Y_j=\sum_{i=1}^{K}\sqrt{P}H_{ij}V_iX_i + N_j
\end{align*}where, $H_{ij}$ denotes the constant channel matrix from Tx-$i$ to Rx-$j$, $V_i$ denotes the precoder at Tx-$i$, $X_i$ denotes the complex symbol vector generated at Tx-$i$, $N_j$ denotes the noise random vector whose coordinates represent independent and identically distributed zero mean unit variance circularly symmetric complex Gaussian random variables. The sizes of the matrices $H_{ij}$, $V_i$, $X_i$, and $N_j$ are given by $n_{r_j}\times n_{t_i}$, $n_{t_i}\times {d_i}$, $ {d_i} \times 1$, and $n_{r_j}\times 1$ respectively, where $n_{r_j}$ and $n_{t_i}$ denote the number of receive and transmit antennas at Rx-$j$ and Tx-$i$ respectively, and ${d_i}$ denotes the number of independent complex symbols per channel use that Tx-$i$ wants to transmit to Rx-$i$. These complex symbols are assumed to take values from finite constellations with uniform distribution over its elements. For simplicity of exposition, we assume that the finite constellations used are all the same at all the transmitters 
and are of cardinality $M$. The results of this paper apply with simple modifications when this is not the case. The finite constellation used is denoted by ${\cal S}$ and is of unit power.

\section{CCSC Optimal Precoders} \label{sec3}
In this section, we derive a set of necessary and sufficient on the precoders for CCSC optimality which is taken to be a measure of optimality for linear precoding in the high SNR regime for the finite constellation input case. Rate achievable for Tx-$i$ by treating interference as noise at Rx-$i$, for all $i$, is given by $R_i<I[X_i;Y_i]$. Our focus will be on the boundary point given by $I[X_i;Y_i]$, for all $i$.  Let $X=[X_1 ~X_2 ~\cdots ~X_K]^T$. The effective channel matrix from all the transmitters to Rx-$i$ is given by $H_i=[H_{1i} ~H_{2i} \cdots ~H_{Ki}]$. Define
\begin{align*}
 V = diag(V_1,V_2,\cdots,V_K).
\end{align*}
Using the chain rule for mutual information \cite{CoT},

{\small
\begin{align*}
 I[X_i;Y_i]=I[X_1,X_2,\cdots,X_K;Y_i]-I[X_1,X_2,\cdots,X_K;Y_i|X_i].
\end{align*}}With uniform distribution assumed over the elements of the constellation, the expression for $I[X_1,X_2,\cdots,X_K;Y_i]$ is given by (\ref{MI_exp}) (at the top of the next page) which is derived in a similar way as in \cite{AbR}. Define the matrix
\begin{align} \label{eqn-define_A}
A^{k_1,k_2}_i=H_iV\left(x^{k_1}-x^{k_2} \right).
\end{align}
The following theorem gives a high SNR approximation for $I[X_1,X_2,\cdots,X_K;Y_i]$.
\begin{theorem}
\label{thm1}
 At high $P$, the mutual information $I[X_1,X_2,\cdots,X_K;Y_i]$ can be approximated by

{\small
\begin{align}
\label{eqn_approx}
 &logM^{\sum_{i=1}^{K}d_k} \\
\nonumber
&-\frac{1}{M^{\sum_{k=1}^{K}d_k}}\hspace{-0.2cm}\sum_{k_1 = 0}^{M^{\sum_{k=1}^{K}d_k}-1}\left[ log  \left( \sum_{k_2 = 0}^{M^{\sum_{k=1}^{K}d_k}-1}e^{ - \left|\left|\sqrt P A^{k_1,k_2}_i\right|\right|^2 }\right) \right].
\end{align}}
\end{theorem}
\begin{proof}
\begin{figure*}
\scriptsize
\begin{align}
\centering
\label{MI_exp}
&I(X_1,X_2,\cdots,X_K;Y_i)= logM^{\sum_{k=1}^{K}d_k} - \frac{1}{M^{\sum_{k=1}^{K}d_k}} \sum_{k_1 = 0}^{M^{\sum_{k=1}^{K}d_k}-1} E_{N_i} \left[log  \sum_{k_2=0}^{M^{\sum_{k=1}^{K}d_k}-1} exp - \left( \left|\left|N_i +\sqrt P H_iV\left( x^{k_1}-x^{k_2} \right)\right|\right|^2-\left|\left|N_i\right|\right|^2 \right) \right]\\
\label{MI_eqn1}
&=logM^{\sum_{k=1}^{K}d_k} - \frac{1}{M^{\sum_{k=1}^{K}d_k}} \sum_{k_1 = 0}^{M^{\sum_{k=1}^{K}d_k}-1} E_{N_i} \left[log  \left(1+|{\cal A}^{k_1}_i|+\sum_{\substack{{k_2 \neq k_1}\\{k_2 \notin {\cal A}^{k_1}_i}}} exp - \left( \left|\left|N_i +\sqrt P A^{k_1,k_2}_i\right|\right|^2-\left|\left|N_i\right|\right|^2 \right) \right)\right]\\
\label{MI_eqn2}
&=logM^{\sum_{k=1}^{K}d_k} - \frac{1}{M^{\sum_{k=1}^{K}d_k}} \sum_{k_1 = 0}^{M^{\sum_{k=1}^{K}d_k}-1} E_{N_j} \left[log  \left(1+|{\cal A}^{k_1}_i|+\sum_{\substack{{k_2 \neq k_1}\\{k_2 \notin {\cal A}^{k_1}_i}}} exp -\left(\sum_{l \notin {\cal A}^{k_1,k_2}_i} \left( \left|N_i(l) +\sqrt P A^{k_1,k_2}_i(l)\right|^2-\left|N_i(l)\right|^2 \right) \right)\right)\right]\\
\nonumber
&=logM^{\sum_{k=1}^{K}d_k} -\frac{1}{M^{\sum_{k=1}^{K}d_k}}\sum_{k_1 = 0}^{M^{\sum_{k=1}^{K}d_k}-1}\left[\int_{\substack{{(n_{iR}(l),n_{iI}(l))}\\{l \notin {\cal A}^{k_1,k_2}_i}}} \prod_{l \notin {\cal A}^{k_1,k_2}_i}p_{N_i(l)}(n_j(l)) \right.\\
\label{MI_eqn3}
&\hspace{5.7cm} \left. \times ~log  \left(1+|{\cal A}^{k_1}_i|+\sum_{\substack{{k_2 \neq k_1}\\{k_2 \notin {\cal A}^{k_1}_i}}} e^{ -\left(\sum_{l \notin {\cal A}^{k_1,k_2}_i} \left( \left|n_i(l) +\sqrt P A^{k_1,k_2}_i(l)\right|^2-\left|n_i(l)\right|^2 \right)\right)} \right) \prod_{l \notin {\cal A}^{k_1,k_2}_i} dn_i(l)\right] \\
\nonumber
&\approx logM^{\sum_{k=1}^{K}d_k} -\frac{1}{M^{\sum_{k=1}^{K}d_k}}\sum_{k_1 = 0}^{M^{\sum_{k=1}^{K}d_k}-1}\left[\int_{\substack{{(n_{iR}(l),n_{iI}(l))}\\{l \notin {\cal A}^{k_1,k_2}_i}}} \prod_{l \notin {\cal A}^{k_1,k_2}_i}p_{N_i(l)}(n_i(l)) \right.\\
\label{MI_eqn4}
&\hspace{5.7cm} \left. \times ~log  \left(1+|{\cal A}^{k_1}_i|+\sum_{\substack{{k_2 \neq k_1}\\{k_2 \notin {\cal A}^{k_1}_i}}} e^{ -\left(\sum_{l \notin {\cal A}^{k_1,k_2}_i} \left( \left|\sqrt P A^{k_1,k_2}_i(l)\right|^2 -\left|n_i(l)\right|^2 \right)\right)} \right) \prod_{l \notin {\cal A}^{k_1,k_2}_i} dn_i(l)\right] \\
\nonumber
&\approx logM^{\sum_{k=1}^{K}d_k} -\frac{1}{M^{\sum_{k=1}^{K}d_k}}\sum_{k_1 = 0}^{M^{\sum_{k=1}^{K}d_k}-1}\left[\int_{\substack{{(n_{jR}(l),n_{jI}(l))}\\{l \notin {\cal A}^{k_1,k_2}_i}}} \prod_{l \notin {\cal A}^{k_1,k_2}_i}p_{N_j(l)}(n_j(l)) \right.\\
\label{MI_eqn5}
&\hspace{5.7cm} \left. \times ~log  \left(1+|{\cal A}^{k_1}_i|+\sum_{\substack{{k_2 \neq k_1}\\{k_2 \notin {\cal A}^{k_1}_i}}} e^{ -\left(\sum_{l \notin {\cal A}^{k_1,k_2}_i} \left( \left|\sqrt P A^{k_1,k_2}_i(l)\right|^2\right)\right)} \right) \prod_{l \notin {\cal A}^{k_1,k_2}_i} dn_j(l)\right] \\
\label{MI_eqn6}
&= logM^{\sum_{k=1}^{K}d_k} -\frac{1}{M^{\sum_{k=1}^{K}d_k}}\sum_{k_1 = 0}^{M^{\sum_{k=1}^{K}d_k}-1}\left[ log  \left(1+|{\cal A}^{k_1}_i|+\sum_{\substack{{k_2 \neq k_1}\\{k_2 \notin {\cal A}^{k_1}_i}}} e^{ -\left(\sum_{l \notin {\cal A}^{k_1,k_2}_i} \left( \left|\sqrt P A^{k_1,k_2}_i(l)\right|^2\right)\right)} \right) \right] \\
\label{MI_eqn7}
&= logM^{\sum_{k=1}^{K}d_k} -\frac{1}{M^{\sum_{k=1}^{K}d_k}}\sum_{k_1 = 0}^{M^{\sum_{k=1}^{K}d_k}-1}\left[ log  \left(1+|{\cal A}^{k_1}_i|+\sum_{\substack{{k_2 \neq k_1}\\{k_2 \notin {\cal A}^{k_1}_i}}} e^{ -\left(\sum_{l=1}^{n_{r_i}} \left( \left|\sqrt P A^{k_1,k_2}_i(l)\right|^2\right)\right)} \right) \right] 
\end{align}
\hrule
\vspace{-0.5cm}
\end{figure*}
Define the sets
\begin{align}
\label{set1}
&{\cal A}^{k_1}_i=\left\{ k_2 \neq k_1 \mid A^{k_1,k_2}_i=0 \right\},\\
\label{set2}
&{\cal A}^{k_1,k_2}_i=\left\{ l \mid A^{k_1,k_2}_i(l)=0 \right\}.
\end{align}The expression in (\ref{MI_exp}) is re-written as in (\ref{MI_eqn1}) and (\ref{MI_eqn2}). At high values of $P$, the following inequality holds good for $k_2 \neq k_1$, $k_2 \notin {\cal A}^{k_1}_i$, and for all $l \notin {\cal A}^{k_1,k_2}_i$.

{\small
\begin{align}
\label{very_greater}
  |\Re\{\sqrt P A^{k_1,k_2}_i(l)\}| >> \frac{3}{\sqrt 2} \text { and } |\Im\{\sqrt P A^{k_1,k_2}_i(l)\}| >> \frac{3}{\sqrt 2}.
\end{align}} Now, note that for $k_2 \neq k_1$, $k_2 \notin {\cal A}^{k_1}_i$, and for all $l \notin {\cal A}^{k_1,k_2}_i$, $\sqrt P A^{k_1,k_2}_i(l)>>n_j(l)$ when $|n_{jR}(l)|\leq 3$ and $|n_{jI}(l)|\leq 3$ where, $n_{jR}(l)$ and $n_{jI}(l)$ represent the real and imaginary parts of $n_j(l)$. The value of $n_{j}(l)$ becomes comparable to $ \sqrt P A^{k_1,k_2}_i(l)$ only if either  $|n_{jR}(l)|> 3$ or $|n_{jI}(l)| > 3$. However, the probability of such an event occurring is extremely small because the variances of the Gaussian random variables $n_{jR}(l)$ and $n_{jI}(l)$ are equal to $\frac{1}{2}$. Hence, the contribution of such an event to the integral in (\ref{MI_eqn3}) is very small that it can be neglected at high $P$. Therefore, on account of (\ref{very_greater}), the approximations in (\ref{MI_eqn4}) and (\ref{MI_eqn5}) are valid. The equation (\ref{MI_eqn6}) follows from the fact that probability distribution integrates to $1$, (\ref{MI_eqn7}) follows from (\ref{set2}), and the proposed approximation 
in (\ref{eqn_approx}) is obtained directly by re-writing (\ref{MI_eqn7}).
\end{proof}

\begin{remark}
Note that we cannot straightforwardly argue that at high powers $P$, the noise $N_i$ in (\ref{MI_exp}) can be neglected. This is because the value of $||N_i||$ can be of the order of $||A^{k_1,k_2}||$. The proof uses the fact that the probability of such an event is very small and hence, can be neglected. The greater the power, the better is the approximation. 
A similar approximation was developed for the SISO case in \cite{AbR}. The approximation given in Theorem \ref{thm1} is a generalization of the approximation derived in \cite{AbR}.
\end{remark}

Let $X_{\not{\hspace{0.05cm}i}}=[X_1 ~X_2 ~\cdots ~X_{i-1} ~X_{i+1} ~\cdots ~X_K]^T$. The channel matrix from all the transmitters, with the exclusion of Tx-$i$, to Rx-$i$ is given by 
\begin{align*}
 H_{\not{\hspace{0.05cm}i}}=[H_{11} ~H_{12} ~\cdots~ H_{i-1, i} ~H_{i+1, i} ~\cdots~ H_{Ki}].
\end{align*}Define the matrices 
\begin{align} \nonumber
&V_{\not{\hspace{0.05cm}i}}=\text{diag}(V_1,V_2,\cdots,V_{i-1},V_{i+1},\cdots,V_K).\\
\label{eqn-define_B}
&B^{i_1,i_2}_i=H_{\not{\hspace{0.05cm}i}}V_{\not{\hspace{0.05cm}i}}\left(x_{\not{\hspace{0.05cm}i}}^{i_1}-x_{\not{\hspace{0.05cm}i}}^{i_2} \right),
\end{align}for $i_1,i_2=0,1,\cdots, M^{\sum_{j \neq i}d_j}-1$. Similar to (\ref{eqn_approx}), we have the following approximation for $I[X_1,X_2,\cdots,X_K;Y_i|X_i]$ at high $P$.

{\small
\begin{align}
\label{eqn_approx2}
 &log~M^{\sum_{\substack{{j\neq i}\\{j=1}}}^{K}d_i} \\
\nonumber
&-\frac{1}{ M^{\sum_{j \neq i}d_j}}\hspace{-0.2cm}\sum_{i_1 = 0}^{ M^{\sum_{j \neq i}d_j}-1}\left[ log  \left( \sum_{i_2 = 0}^{ M^{\sum_{j \neq i}d_j}-1}e^{ - \left|\left|\sqrt P B^{i_1,i_2}_i\right|\right|^2 }\right) \right]
\end{align}}Hence, a high SNR approximation for $I[X_i;Y_i]$ is given by

{\small
\begin{align}
\nonumber
&I[X_i;Y_i] \approx log~M^{d_i} \\
\label{actual_approx}
&-\frac{1}{M^{\sum_{k=1}^{K}d_k}}\hspace{-0.2cm}\sum_{k_1 = 0}^{M^{\sum_{k=1}^{K}d_k}-1}\left[ log  \left( \sum_{k_2 = 0}^{M^{\sum_{k=1}^{K}d_k}-1}e^{ - \left|\left|\sqrt P A^{k_1,k_2}_i\right|\right|^2 }\right) \right].\\
\nonumber
&+\frac{1}{ M^{\sum_{j \neq i}d_j}}\hspace{-0.2cm}\sum_{i_1 = 0}^{ M^{\sum_{j \neq i}d_j}-1}\left[ log  \left( \sum_{i_2 = 0}^{ M^{\sum_{j \neq i}d_j}-1}e^{ - \left|\left|\sqrt P B^{i_1,i_2}_i\right|\right|^2 }\right) \right].
\end{align}}
Now, define the set 
\begin{align}
&{\cal B}^{i_1}=\left\{ i_2 \neq i_1 \mid B^{i_1,i_2}_i=0 \right\}.
\end{align}
The following theorem gives a set of necessary and sufficient conditions under which the above approximation tends to $log~M^{d_i}$ as $P$ tends to infinity and hence, gives a set of necessary and sufficient conditions under which the precoders are CCSC optimal.

\begin{theorem} \label{thm2}
 The approximation for $I[X_i;Y_i]$ given in (\ref{actual_approx}) tends to $log~M^{d_i}$ as $P$ tends to infinity iff

{\vspace{-0.4cm} \footnotesize
\begin{align}
\label{CCSC-optimality}
 &H_{ii}V_i(x_i^{p_{i1}}-x_i^{p_{i2}})+\sum_{\substack{{k \neq i}\\{k=1}}}^{K}H_{kj}V_k(x_k^{p_{k1}}-x_k^{p_{k2}}) \neq \underline{0}, \\
\nonumber
&~~~~\forall  ~p_{i1} \neq p_{i2}, ~\forall  ~p_{k1},~\forall p_{k2}
\end{align}}where, $p_{k1},p_{k2}=0,1,\cdots,M^{d_l}-1$.
\end{theorem}
\begin{proof}
The summation-term of the second term in (\ref{actual_approx}) is re-written as

{\footnotesize
\begin{align}
\nonumber
&\sum_{k_1 = 0}^{M^{\sum_{k=1}^{K}d_k}-1} \left[ log  \left(1+|{\cal A}^{k_1}_i|+\sum_{\substack{{k_2 \neq k_1}\\{k_2 \notin {\cal A}^{k_1}_i}}} e^{ -P\left( \left|\left| A^{k_1,k_2}_i\right|\right|^2\right)} \right) \right]\\
\label{eqn-use_in_dmin_1}
&=\sum_{k_1 = 0}^{M^{\sum_{k=1}^{K}d_k}-1} log  \left(1+|{\cal A}^{k_1}_i| \right)+\\
\nonumber
&\hspace{2cm}log  \left(1+\frac{1}{1+|{\cal A}^{k_1}_i|}\sum_{\substack{{k_2 \neq k_1}\\{k_2 \notin {\cal A}^{k_1}_i}}} e^{ -P\left( \left|\left| A^{k_1,k_2}_i\right|\right|^2\right)} \right)\\
\nonumber
&\xrightarrow{P \rightarrow \infty} \sum_{k_1 = 0}^{M^{\sum_{k=1}^{K}d_k}-1} log  \left(1+|{\cal A}^{k_1}_i| \right).
\end{align}}Similarly, the summation-term of the last term in (\ref{actual_approx}) tends to $ \sum_{i_1 = 0}^{M^{\sum_{j \neq i}d_j}} log  \left(1+|{\cal B}^{i_1}_i| \right)$ as $P$ tends to infinity. Hence, as $P$ tends to infinity, (\ref{actual_approx}) tends to

{\footnotesize
\begin{align}
\label{eqn-in-terms-of-set}
& log~M^{d_i}-\frac{1}{M^{\sum_{k=1}^{K}d_k}}\sum_{k_1 = 0}^{M^{\sum_{k=1}^{K}d_k}-1} log  \left(1+|{\cal A}^{k_1}_i| \right)\\
\nonumber
&\hspace{2cm}+ \frac{1}{ M^{\sum_{j \neq i}d_j}}\sum_{i_1 = 0}^{M^{\sum_{j \neq i}d_j}} log  \left(1+|{\cal B}^{i_1}_i| \right).
\end{align}}Now, define the sets 

{\footnotesize
\begin{align}
\nonumber
& {\cal C}^{p_{11},\cdots,p_{K1}}_i=\{~(p_{12},p_{22},\cdots,p_{K2})\neq(p_{11},p_{21},\cdots,p_{K1})\mid \\
\nonumber
&H_{ii}V_i(x_i^{p_{i1}}-x_i^{p_{i2}})+\sum_{\substack{{k \neq i}\\{k=1}}}^{K}H_{kj}V_k(x_k^{p_{k1}}-x_k^{p_{k2}}) = \underline{0} ~\},\\
\nonumber
& {\cal D}^{p_{11},\cdots,p_{i-1,1},p_{i+1,1},\cdots, p_{K1}}_i\\
\label{eqn-define_calD}
&=\{~({p_{12},\cdots,p_{i-1,2},p_{i+1,2},\cdots, p_{K2}})\\
\nonumber
&~~~~~~\neq({p_{11},\cdots,p_{i-1,1},p_{i+1,1},\cdots, p_{K1}})\mid \\
\nonumber
&~~~~~~\sum_{\substack{{k \neq i}\\{k=1}}}^{K}H_{kj}V_k(x_k^{p_{k1}}-x_k^{p_{k2}}) = \underline{0} ~\}
\end{align}}where, $p_{l1},p_{l2}=0,1,\cdots,M^{d_l}-1$. Observe that the set of all ${\cal C}^{p_{11},\cdots,p_{K1}}_i$ has a one-one correspondence with the set of all ${\cal A}^{k_1}_i$, and the set of all ${\cal D}^{p_{11},\cdots,p_{i-1,1},p_{i+1,1},\cdots, p_{K1}}_i$ has a one-one correspondence with the set of all ${\cal B}^{i_1}_i$. Hence, (\ref{eqn-in-terms-of-set}) can be re-written as

{\footnotesize
\begin{align}
\nonumber
& log~M^{d_i}-\frac{1}{M^{\sum_{k=1}^{K}d_k}}\sum_{p_{11} = 0}^{M^{d_1}-1} \cdots \sum_{p_{K1} = 0}^{M^{d_K}-1} log  \left(1+|{\cal C}^{p_{11},\cdots,p_{K1}}_i| \right)\\
\label{eqn-in-terms-of-set-1}
&+ \frac{1}{ M^{\sum_{j \neq i}d_j}}\sum_{p_{11} = 0}^{M^{d_1}-1} \cdots \sum_{p_{i-1,1} = 0}^{M^{d_{i-1}}-1}\sum_{p_{i+1,1} = 0}^{M^{d_{i+1}}-1}\cdots \sum_{p_{K1} = 0}^{M^{d_K}-1} \\
\nonumber
&\hspace{3cm}log  \left(1+|{\cal D}^{p_{11},\cdots,p_{i-1,1},p_{i+1,1},\cdots, p_{K1}}_i| \right).
\end{align}}
Now, the set ${\cal C}^{p_{11},\cdots,p_{K1}}_i$ can be written as a disjoint union of two sets, i.e., 
\begin{align}
\label{eqn-union}
{\cal C}^{p_{11},\cdots,p_{K1}}_i={\cal C}^{p_{11},\cdots,p_{K1}}_{1i} \cup {\cal C}^{p_{11},\cdots,p_{K1}}_{2i} 
\end{align}
where, ${\cal C}^{p_{11},\cdots,p_{K1}}_{1i}$ and ${\cal C}^{p_{11},\cdots,p_{K1}}_{2i}$ are defined in (\ref{eqn-set-union1}) and (\ref{eqn-set-union2}) respectively (given at the top of the next page). The set ${\cal C}^{p_{11},\cdots,p_{K1}}_{1i}$ can be re-defined as in (\ref{eqn-set-union1a}). Since {\small $ \sum_{\substack{{k \neq i}\\{k=1}}}^{K}H_{kj}V_k(x_k^{p_{k1}}-x_k^{p_{k2}})$} is independent of the indices $p_{i1}$ and $p_{i2}$, the set ${\cal C}^{p_{11},\cdots,p_{K1}}_{1i}$ is the same for all $p_{i1}=0,1,\cdots,M^{d_i}-1$.
 \begin{figure*}
\scriptsize
\begin{align}
\label{eqn-set-union1}
 &{\cal C}^{p_{11},\cdots,p_{K1}}_{1i} = \{~(p_{12},p_{22},\cdots,p_{K2})\neq(p_{11},p_{21},\cdots,p_{K1}) \text{ for } p_{i2}=p_{i1} \mid  H_{ii}V_i(x_i^{p_{i1}}-x_i^{p_{i2}})+\sum_{\substack{{k \neq i}\\{k=1}}}^{K}H_{kj}V_k(x_k^{p_{k1}}-x_k^{p_{k2}}) =\underline{0} ~\},\\
\label{eqn-set-union1a}
&\Rightarrow{\cal C}^{p_{11},\cdots,p_{K1}}_{1i} = \{~(p_{12},p_{22},\cdots,p_{K2})\neq(p_{11},p_{21},\cdots,p_{K1}) \text{ for } p_{i2}=p_{i1} \mid  \sum_{\substack{{k \neq i}\\{k=1}}}^{K}H_{kj}V_k(x_k^{p_{k1}}-x_k^{p_{k2}}) = \underline{0} ~\}\\
\label{eqn-set-union2}
&{\cal C}^{p_{11},\cdots,p_{K1}}_{2i} = \{~(p_{12},p_{22},\cdots,p_{K2})\neq(p_{11},p_{21},\cdots,p_{K1}) \text{ for }  p_{i2}\neq p_{i1} \mid   H_{ii}V_i(x_i^{p_{i1}}-x_i^{p_{i2}})+\sum_{\substack{{k \neq i}\\{k=1}}}^{K}H_{kj}V_k(x_k^{p_{k1}}-x_k^{p_{k2}}) = \underline{0} ~\}
\end{align}
\hrule
\end{figure*}Now, note that the set of all ${\cal C}^{p_{11},\cdots,p_{i-1,1},0,p_{i+1,1},\cdots,p_{K1}}_1$ has a one-one correspondence with the set of all ${\cal D}^{p_{11},\cdots,p_{i-1,1},p_{i+1,1},\cdots, p_{K1}}_i$ which follows from the definitions of the respective sets. Hence, (\ref{eqn-in-terms-of-set-1}) can be re-written as

{\footnotesize
\begin{align}
 \nonumber
& log~M^{d_i}\\
 \nonumber
&-\frac{1}{M^{\sum_{k=1}^{K}d_k}}\sum_{p_{11} = 0}^{M^{d_1}-1} \cdots \sum_{p_{i-1,1} = 0}^{M^{d_{i-1}}-1}\sum_{p_{i,1} = 0}^{M^{d_{i}}-1}\sum_{p_{i+1,1} = 0}^{M^{d_{i+1}}-1} \cdots \sum_{p_{K1} = 0}^{M^{d_K}-1} \\
\nonumber
&~~~log  \left(1+|{\cal D}^{p_{11},\cdots,p_{i-1,1},p_{i+1,1},\cdots, p_{K1}}_i|+|{\cal C}^{p_{11},\cdots,p_{K1}}_{2i}| \right)\\
 \nonumber
&+ \frac{1}{ M^{\sum_{j \neq i}d_j}}\sum_{p_{11} = 0}^{M^{d_1}-1} \cdots \sum_{p_{i-1,1} = 0}^{M^{d_{i-1}}-1}\sum_{p_{i+1,1} = 0}^{M^{d_{i+1}}-1}\cdots \sum_{p_{K1} = 0}^{M^{d_K}-1} \\
\label{eqn-use_in_dmin_2}
&\hspace{3cm}log  \left(1+|{\cal D}^{p_{11},\cdots,p_{i-1,1},p_{i+1,1},\cdots, p_{K1}}_i| \right)
\end{align}}
{\footnotesize
\begin{align*}
&= log~M^{d_i}\\
&-\frac{1}{M^{\sum_{k=1}^{K}d_k}}\sum_{p_{11} = 0}^{M^{d_1}-1} \cdots \sum_{p_{i-1,1} = 0}^{M^{d_{i-1}}-1}\sum_{p_{i,1} = 0}^{M^{d_{i}}-1}\sum_{p_{i+1,1} = 0}^{M^{d_{i+1}}-1} \cdots \sum_{p_{K1} = 0}^{M^{d_K}-1} \\
&~~~log  \left(1+|{\cal D}^{p_{11},\cdots,p_{i-1,1},p_{i+1,1},\cdots, p_{K1}}_i|+|{\cal C}^{p_{11},\cdots,p_{K1}}_{2i}| \right)\\
&+ \frac{1}{M^{\sum_{k=1}^{K}d_k}}\sum_{p_{11} = 0}^{M^{d_1}-1} \cdots \sum_{p_{i-1,1} = 0}^{M^{d_{i-1}}-1}\sum_{p_{i,1} = 0}^{M^{d_{i}}-1}\sum_{p_{i+1,1} = 0}^{M^{d_{i+1}}-1}\cdots \sum_{p_{K1} = 0}^{M^{d_K}-1} \\
\nonumber
&\hspace{3cm}log  \left(1+|{\cal D}^{p_{11},\cdots,p_{i-1,1},p_{i+1,1},\cdots, p_{K1}}_i| \right).
\end{align*}}Clearly, if $|{\cal C}^{p_{11},\cdots,p_{K1}}_{2i}|>1$ for some $(p_{11},\cdots,p_{K1})$ then, the second term in the above equation is strictly greater than the last term and hence, as $P$ tends to infinity, $I[X_i;Y_i]$ tends to a value that is strictly less than $log~M^{d_i}$. If  $|{\cal C}^{p_{11},\cdots,p_{K1}}_{2i}|=0$ for all $p_{l1}$ then, the second term in the above equation is equal to the last term and hence, $I[X_i;Y_i]$ tends to a value equal to $log~M^{d_i}$. Thus, $I[X_i;Y_i]$ tends to $log~M^{d_i}$ as $P$ tends to infinity iff (\ref{CCSC-optimality}) is satisfied.
\end{proof}

\begin{remark}
 The result of Theorem \ref{thm2} means that the rate achieved by treating interference as noise at high $P$ tends to CCSC for Tx-$i$ iff, in the absence of the Gaussian noise, two different symbol vectors $x^{p_{i1}}_i$ and $x^{p_{i2}}_i$ sent by Tx-$i$ should not map to the same symbol vector at Rx-$i$ for any data symbol transmitted by the interfering transmitters.
\end{remark}

\begin{remark}
For a given value of channel gains with none of the direct channel gains being $0$, when the entries of the precoders are chosen from any continuous distribution (say, standard normal distribution) the probability of the event

{\vspace{-0.4cm} \footnotesize
\begin{align*}
H_{ii}V_i(x_i^{p_{i1}}-x_i^{p_{i2}})+\sum_{\substack{{k \neq i}\\{k=1}}}^{K}H_{kj}V_k(x_k^{p_{k1}}-x_k^{p_{k2}}) = \underline{0},
\end{align*}}to occur for any $p_{i1} \neq p_{i2}$ and for any $(p_{k1},~p_{k2})$ is zero. By appropriate scaling of the precoders thus obtained, with probability $1$, we have CCSC optimal precoders.
\end{remark}

\begin{remark}
Interference alignment, if feasible \cite{GSB} for the given values of $n_{t_i}$, $n_{r_i}$, and $d_i$ involves finding precoders such that the signal sub-space at Rx-$i$, generated by $[H_{ii}V_i]$, is linearly independent of the interference sub-space, generated by $[H_{1i}V_1 ~\cdots ~H_{i-1,i}V_{i-1} ~H_{i+1,i}V_{i+1} ~\cdots ~H_{K,i}V_{K}]$, and the matrix $[H_{ii}V_i]$ is full-rank \cite{CaJ1}. The CCSC optimality condition in (\ref{CCSC-optimality}) can be rewritten as

{\footnotesize
\begin{align*}
 [H_{ii}V_i ~~H_{1i}V_1 ~\cdots ~H_{i-1,i}V_{i-1} ~~H_{i+1,i}V_{i+1} ~\cdots ~H_{K,i}V_{K}]\\
\times
\begin{bmatrix}
&(x_i^{p_{i1}}-x_i^{p_{i2}})\\
&(x_1^{p_{11}}-x_1^{p_{12}})\\
&\vdots\\
&(x_{i-1}^{p_{i-1,1}}-x_{i-1}^{p_{i-1,2}})\\
&(x_{i+1}^{p_{i+1,1}}-x_{i+1}^{p_{i+1,2}})\\
&\vdots\\
&(x_K^{p_{K1}}-x_K^{p_{K2}})
\end{bmatrix}\neq \underline{0}, ~\forall  ~p_{i1} \neq p_{i2}.
\end{align*}}Since, with precoders that achieve IA, the signal sub-space at Rx-$i$ is linearly independent of the interference sub-space and $[H_{ii}V_i]$ is full-rank, the above condition is satisfied for all $i$. Hence, IA precoders are also CCSC optimal precoders. However, in general, finding such precoders are NP-hard \cite{RSL} whereas finding CCSC optimal precoders are easy to find as explained in the previous remark.
\end{remark}
 \begin{figure}[htbp]
\centering
\includegraphics[totalheight=3.1in,width=3.6in]{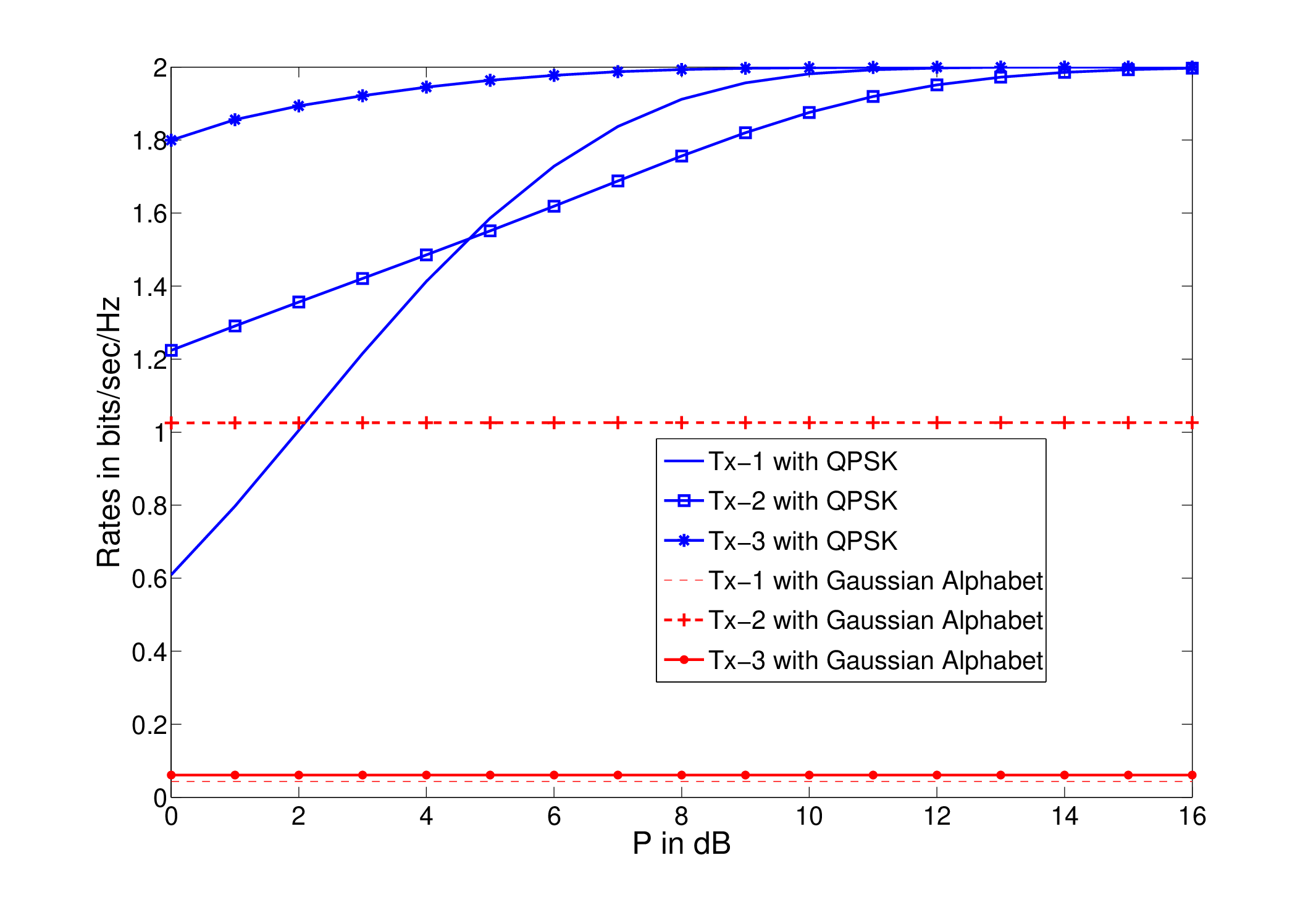}
\caption{Rates in bits/sec/Hz vs $P$ in dB for Example \ref{eg1}.}
\label{fig1}
\end{figure}

The following example illustrates a $3$-user MIMO GIC which employs CCSC optimal precoders.
\begin{example} \label{eg1}
Consider a MIMO GIC with $K=3$, $n_{t_i}=n_{r_i}=1$, $d_i=1$ for all $i$, and the finite constellation used is QPSK. The channel matrix and the precoders are given by
\begin{align*}
&H= \begin{bmatrix}
    -0.9+0.4i &  -1.7-1.40i &   1.5 + 5.0i\\
    2.6-0.9i  &  -0.9-2.8i  & 0.04 + 0.88i\\
   -2.9-5.2i  &  -10.2+0.7i & -0.5 + 2.4i]
    \end{bmatrix},\\
& V_1=1, V_2=e^{i\frac{\pi}{3}}, V_3=1
\end{align*}where, the matrix element $[H]_{ij}$ represents the channel gain from Rx-$i$ to Tx-$j$. The mutual information $I[X_i;Y_i]$ evaluated using Monte-Carlo simulation is plotted for QPSK inputs and Gaussian inputs, for all $i$, in Fig. \ref{fig1}. The chosen precoders satisfy (\ref{CCSC-optimality}) and hence, $I[X_i;Y_i]$ saturates to $2$ bits/sec/Hz for all $i$, as $P$ tends to infinity in the QPSK case whereas, for the Gaussian alphabet case, the saturation rate is determined by the channel gains. The saturation value of $I[X_i;Y_i]$ in the Gaussian alphabet case is given by $log\left(1+\frac{|h_{ii}|^2}{\sum_{k\neq i}^{}|h_{ki}|^2}\right)$ which evaluates to $0.04$, $1.02$, and $0.06$ bits/sec/Hz for Tx-$1$, Tx-$2$ and Tx-$3$ respectively in this example.
\end{example}
\begin{figure}[htbp]
\centering
\includegraphics[totalheight=3.1in,width=3.6in]{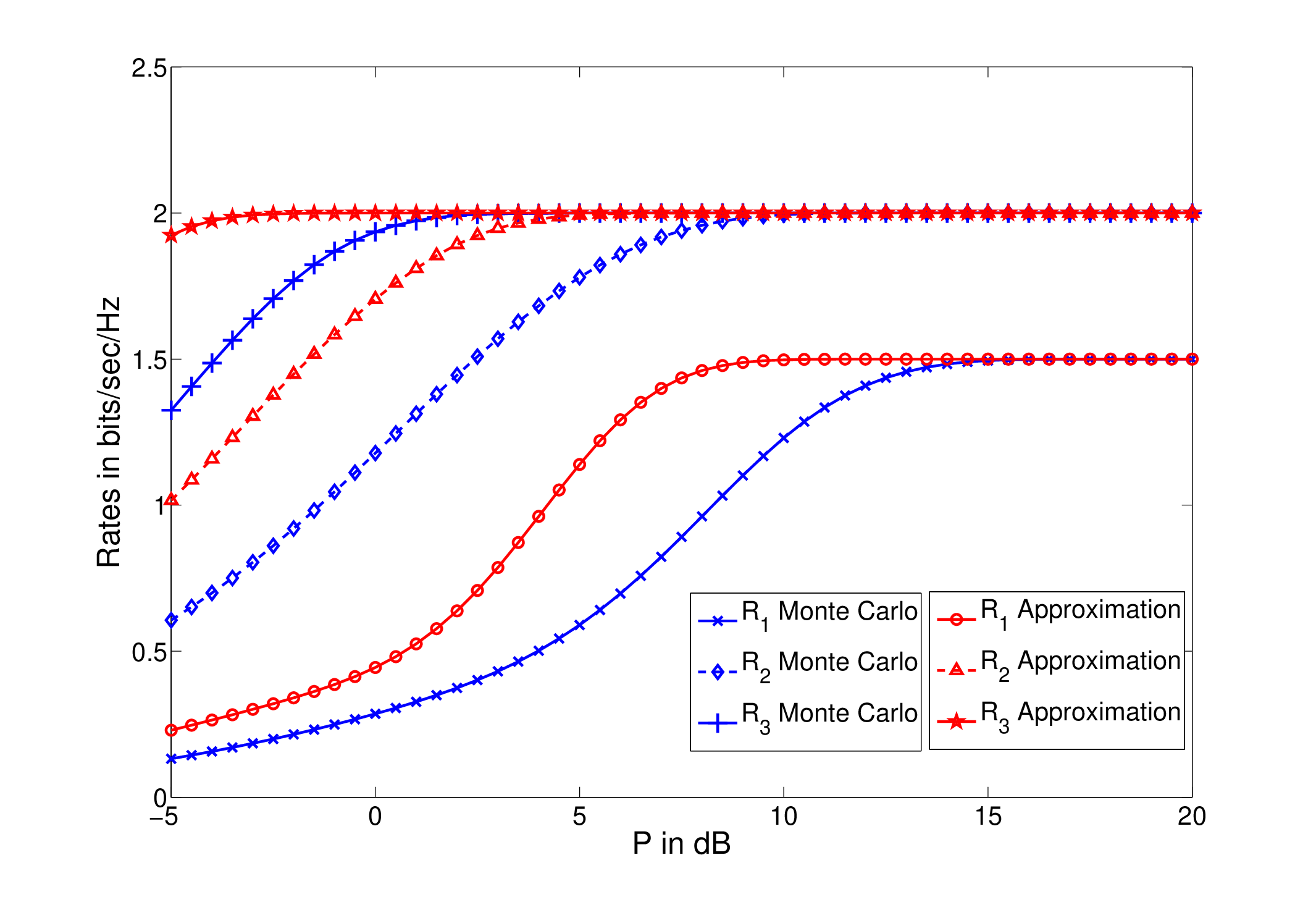}
\caption{Rates in bits/sec/Hz vs $P$ in dB for Example \ref{eg1a}.}
\label{fig1a}
\end{figure}

The following example illustrates a naive choice of precoders that are not CCSC optimal in a $3$-user MIMO GIC.
\begin{example} \label{eg1a}
Consider a MIMO GIC with $K=3$, $n_{t_i}=n_{r_i}=2$, $d_i=1$ for all $i$, and the finite constellation used is QPSK. Let the effective matrix from all the transmitters to all the receivers be given by
\begin{align*}
H=\begin{bmatrix}
    H_{11} & H_{21} & H_{31}\\
    H_{12} & H_{22} & H_{32}\\
    H_{13} & H_{23} & H_{33}
  \end{bmatrix}
\end{align*}where, $H_{ij}$ is the $2\times 2$ channel matrix from Tx-$i$ to Rx-$j$. The precoders and the channel matrices are given in (\ref{eqn-prec}) and (\ref{eqn-ch}) respectively (given at the top of the next page).
\begin{figure*}
\scriptsize
\begin{align}
\label{eqn-prec}
&V_1= \begin{bmatrix}
	0.66 + 0.74i& \\
	0.13 + 0.99i&
     \end{bmatrix}, 
V_2= \begin{bmatrix}
	0.9883 + 0.1524i& \\
	0.4538 + 0.8911i&
     \end{bmatrix}, 
V_3=\begin{bmatrix}
	0.7044 + 0.7098i& \\
	0.1603 + 0.9871i&
     \end{bmatrix}, \\
\label{eqn-ch}
&H= \begin{bmatrix}
   0.5756 - 0.0565i &  0.7524 - 0.1375i &  0.1697 - 0.1069i &  0.0124 - 0.2002i  &        0         &	     0 	\\
   0.1610 + 0.3766i & -0.0010 + 0.2005i &  0.8758 - 0.0689i & -0.1285 + 0.0605i  &        0         &        0   \\
  -1.1533 - 0.1280i & -0.6361 + 1.4658i & -1.3069 + 0.1090i &  0.0427 + 0.2488i & -0.0028 + 0.2215i & -1.0597 - 0.2708i\\
  -1.7763 - 0.3748i &  0.5341 + 0.0966i & -0.9491 + 0.8074i & -1.0773 - 1.7202i &  0.9616 - 1.2130i & -0.6077 + 0.6970i\\
  -1.7082 - 0.4948i & -0.6101 - 0.4739i & -0.2226 - 4.2486i & -0.8216 + 0.4808i &  0.9572 + 1.8870i & -1.4428 - 1.4353i\\
  -1.3014 - 0.5614i &  1.2515 + 0.3414i &  0.4242 + 0.0202i &  0.0138 - 0.8740i &  0.3393 - 1.3451i &  0.9498 - 1.0932i
    \end{bmatrix},
\end{align}
\hrule
\end{figure*}
Note that the QPSK points are given by  $\left(\frac{1}{\sqrt 2} + \frac{1}{\sqrt 2}i, -\frac{1}{\sqrt 2}+\frac{1}{\sqrt 2}i,  -\frac{1}{\sqrt 2}-\frac{1}{\sqrt 2}i, \frac{1}{\sqrt 2}-\frac{1}{\sqrt 2}i \right)$. Now, (\ref{CCSC-optimality}) is not satisfied for $i=1$ because 

{\footnotesize
\begin{align*}
&[H_{11}V_1 ~H_{21}V_2 ~H_{31}V_3] \times\begin{bmatrix}
&(x_1^{p_{11}}-x_1^{p_{12}})\\
&(x_{2}^{p_{21}}-x_{2}^{p_{22}})\\
&(x_{3}^{p_{31}}-x_{3}^{p_{32}})
\end{bmatrix}=\underline{0}, ~~\text{for}\\
&\begin{bmatrix}
&(x_1^{p_{11}}-x_1^{p_{12}})\\
&(x_{2}^{p_{21}}-x_{2}^{p_{22}})\\
&(x_{3}^{p_{31}}-x_{3}^{p_{32}})
\end{bmatrix}=
\begin{bmatrix}
\sqrt 2+\sqrt 2i\\
\sqrt 2\\
0
\end{bmatrix}.
\end{align*}}For $i=2,3$, (\ref{CCSC-optimality}) is satisfied. The plots of $I[X_i; Y_i]$ evaluated using Monte-Carlo simulation and $I[X_i; Y_i]$ evaluated using the high SNR approximation in (\ref{actual_approx}) is shown in Fig. \ref{fig1a}. Note that $R_1$ saturates to a value strictly less than $2$ bits/sec/Hz whereas $R_2$ and $R_3$ saturates to $2$ bits/sec/Hz, thus validating Theorem \ref{thm2}.
\end{example}

The intuition behind the result of Theorem \ref{thm2} is as follows. Define the sum constellation at receiver Rx-$i$ to be the set of points given by 
\begin{align*}
\left\{\sqrt{P}\left(H_{ii}V_i X_i+\sum_{\substack{{k \neq i}\\{k=1}}}^{K}H_{kj}V_k X_k\right)  \mid  X_j \in {\cal S}, ~\forall ~j \right\}.
\end{align*}

At every receiver Rx-$i$, the interference forms a ``cloud'' around the desired signal points in the sum constellation. Cloud around a desired signal point $x_i^{p_{i}}$, where $x_i^{p_{i}} \in {\cal S}$, is defined as the set of points given by
\begin{align*}
\left\{\sqrt{P}\left(H_{ii}V_ix_i^{p_{i}}+\sum_{\substack{{k \neq i}\\{k=1}}}^{K}H_{kj}V_k X_k\right) \mid X_k \in {\cal S}\right\}.
\end{align*}Note that the information regarding a desired signal point $x_i^{p_{i}}$ is contained in its respective cloud in the sum constellation. At high values of $P$, the clouds corresponding to the different signal points move away from each other if there is no intersection among the clouds. Since it is enough for each receiver to distinguish between the clouds and is not required to distinguish the points inside every cloud, every signal point can be reliably decoded if the clouds do not intersect. The sum constellation and the non-intersecting clouds at Rx-$1$, at $P=16~dB$, for Example \ref{eg1} is plotted in Fig. \ref{fig-sum_const}. 
\begin{figure}[htbp]
\centering 
\includegraphics[totalheight=2.5in,width=3.6in]{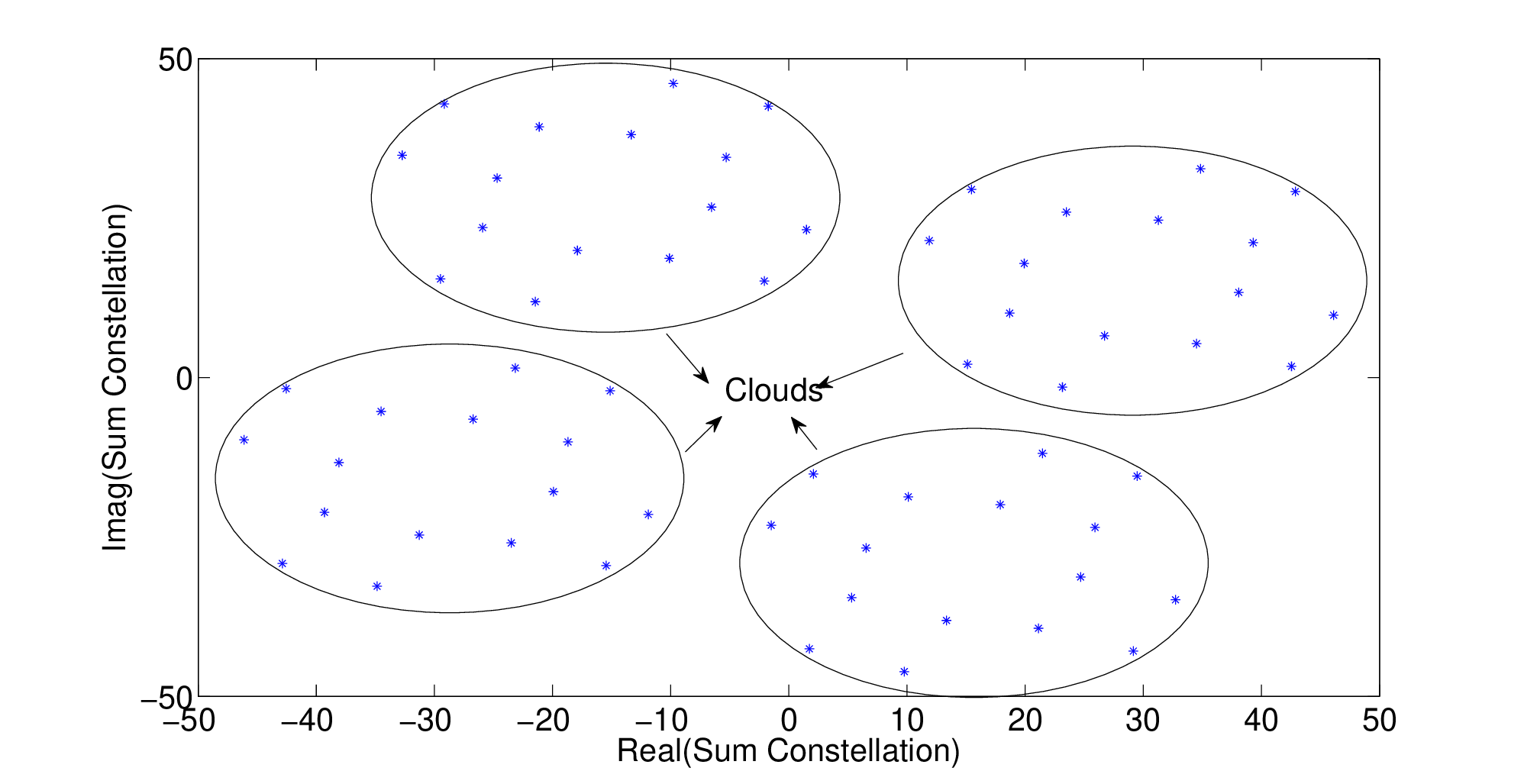}
\caption{Cloud Constellation at Rx-$1$ for Example \ref{eg1} at $P=16~dB$.}
 \label{fig-sum_const}
\end{figure} 

\section{Gradient Ascent Based Algorithm for Finite-SNR} \label{sec4}
In the previous section, we studied the rate achieved for each transmitter by treating interference as noise at every receiver as $P$ tends to infinity. In this section, we focus on the finite SNR case. Specifically, the aim is to maximize the sum-rate achieved by treating interference as noise at every receiver with respect to the precoders, i.e.,

{\small
\begin{align*}
 \max f(V_1,\cdots,V_K)= \max \sum_{i=1}^{K}I[X_i;Y_i]  \text{ with } Tr (V_iV_i^H) \leq 1.
\end{align*}}This is a non-concave problem in general and  difficult to solve. Hence, we propose a gradient-ascent based algorithm to improve the sum-rate starting from some random initialization of precoders. Define the MMSE matrix at Rx-$j$ by
\begin{align*}
  E_j = \mathbb E[(X-\mathbb E[X|Y_j])(X-\mathbb E[X|Y_j])^H]
\end{align*}where, $\mathbb E$ represents the expectation operator. Define the MMSE matrix at Rx-$j$ with the exclusion of Tx-$j$'s signal by

{\footnotesize
\begin{align*}
  E_{{\not{\hspace{0.02cm}j}}} = \mathbb E[(X_{\not{\hspace{0.02cm}j}}-\mathbb E[X_{\not{\hspace{0.02cm}j}}|Y_j-H_{jj}X_j])(X_{\not{\hspace{0.02cm}j}}-\mathbb E[X_{\not{\hspace{0.02cm}j}}|Y_j-H_{jj}X_j])^H]
\end{align*}}The gradient of the sum-rate with respect to the precoder $V_i$ given by 

{\small
\begin{align}
\nonumber
&\nabla_{V_i} f(V_1,\cdots,V_K)= \nabla_{V_i}  \sum_{j=1}^{K}I[X_j;Y_j]\\
\nonumber
&=\nabla_{V_i}  \sum_{j=1}^{K}I[X_1,X_2,\cdots,X_K;Y_j]-I[X_1,X_2,\cdots,X_K;Y_j|X_j]\\
\label{eqn-grad2}
&=log~e\sum_{j=1}^{K} H^{H}_{ij}H_jE_j~I_{\sum_{k=1}^{i-1}d_k+1:\sum_{k=1}^{i}d_k}\\
\nonumber
&-log~e\sum_{\substack{{j=1}\\{j\neq i}}}^{i=K} H^{H}_{ij}H_{\not{\hspace{0.02cm}j}}E_{\not{\hspace{0.02cm}j}}~I_{\sum_{k=1}^{i-1}d_k-d_j{\cal I}(i-j)+1:\sum_{k=1}^{i}d_k-d_j{\cal I}(i-j)}
\end{align}}where, (\ref{eqn-grad2}) follows from the relation between the gradient of mutual information and the MMSE matrix obtained in \cite{PaV}. The matrices 

{\small \vspace{-0.3cm}
\begin{align*} 
&I_{\sum_{k=1}^{i-1}d_k+1:\sum_{k=1}^{i}d_k} \text{ and }\\
&I_{\sum_{k=1}^{i-1}d_k-d_j{\cal I}(i-j)+1:\sum_{k=1}^{i}d_k-d_j{\cal I}(i-j)} 
\end{align*}}select the column numbers from {\small$\sum_{k=1}^{i-1}d_k+1$} to {\small$\sum_{k=1}^{i}d_k$} of $E_j$ and {\small$\sum_{k=1}^{i-1}d_k-d_j{\cal I}(i-j)+1$} to {\small$\sum_{k=1}^{i}d_k-d_j{\cal I}(i-j)$} of $E_{\not{\hspace{0.02cm}j}}$ respectively, where

{\small \[
{\cal I}(i-j)=\left\{
\begin{array}{ll}
1& ~i>j \\
0 & ~i<j.
\end{array}
\right.
\]}Define $V=diag(V_1,V_2,\cdots,V_K)$. The gradient ascent based algorithm for optimizing $f(V_1,\cdots,V_K)$ with respect to the precoders is given in Algorithm $1$. During every iteration, whose number is denoted by $n$, all the precoders are updated as given in Line $10$ of Algorithm $1$ where, $\nabla_V f|_{V=V^{(n-1)}}$ represents {\small$diag\left(\nabla_{V_1} f|_{V_1=V_1^{(n-1)}},\nabla_{V_2} f|_{V_2=V_2^{(n-1)}},\cdots,\nabla_{V_K} f|_{V_K=V_K^{(n-1)}}\right)$} and {\small$\nabla_{V_i} f|_{V_i=V_i^{(n-1)}}$} denotes the gradient {\small$\nabla_{V_i} f$} evaluated at {\small$V_i=V_i^{(n-1)}$}. If the power constraint for any transmitter Tx-$i$ is violated then, $V_i^{(n)}$ is projected onto the feasible set with {\small$Tr(V_i^{(n)}{V_i^{(n)}}^H)\leq 1$} (see Line $12$ of Algorithm $1$) \cite{PaV}. The condition in Line $15$ of the algorithm ensures that there is sufficient increase in the objective function. The step size $t$ of the algorithm is chosen by back-tracking line search with parameters $\alpha$ and $\beta$ whose typical values lie between $(0.01,0.3)$ and $(0.1,0.8)$\cite{BoV}. The proposed algorithm stops when either the number of iterations performed is equal to $max\_iterations$ or $f^{(n-1)}-f^{(n-2)}<\epsilon$ (see Line $5$ of Algorithm $1$), for some fixed $\epsilon$.

\begin{algorithm}
 \label{algo1}
 \footnotesize
 \begin{algorithmic}[1]
 \State Initialize $V_i=V_i^{(0)}$ with $Tr(V_iV_i^H)\leq 1$, $i=1,2,\cdots, K$, and $t=1$.
 \For {$n=1$ to $\text{\em max\_iterations}$}
     \State \textbf{Compute} $f^{(n-1)}=f(V^{(n-1)}_1,V^{(n-1)}_2,\cdots,V^{(n-1)}_K)$, \State$~~~~~~~~~~~$ $E^{(n-1)}_j$, and $E^{(n-1)}_{\not{\hspace{0.02cm}j}}$, for $j=1,2,\cdots,K$.
      \If{$n>1$ \text{and} $f^{(n-1)}-f^{(n-2)}<\epsilon$}{\\ ~~~~~~~exit \textbf{for}}
      \EndIf
    \State Compute $\nabla_V f|_{V=V^{(n-1)}}$
    \State \textbf{do}
    \State $~~$ $V^{(n)} \leftarrow V^{(n-1)}+t\nabla_V f|_{V=V^{(n-1)}}$.\\
    \State $~~$ $V_i^{(n)}\leftarrow \frac{V_i^{(n)}}{Tr(V_i^{(n)}{V_i^{(n)}}^H)}$ if $Tr(V_i^{(n)}{V_i^{(n)}}^H)>1$, for all $i$.
    \State $~~$ Compute $f^{(n)}=f(V^{(n)}_1,V^{(n)}_2,\cdots,V^{(n)}_K)$.
    \State $~~$ $t=\beta t$
    \State \textbf{while} $f^{(n)} < f^{(n-1)}+\alpha t \left|\left|\nabla_V f|_{V=V^{(n-1)}}\right|\right|_F^2$
    \State $t=1$.
\EndFor
 \end{algorithmic}
 \caption{Gradient Ascent based Algorithm for improving sum-rate}
\end{algorithm}
Similar gradient ascent based algorithms have been proposed in the past for optimizing rates in single user MIMO channels \cite{PaV,PRS} and MIMO MAC \cite{WZX} with finite constellation inputs and precoding. Like in \cite{PaV}\cite{PRS}\cite{WZX}, the algorithm does not assume uniform distribution over the elements of the finite constellation. The above algorithm appeared first in an older \textit{arxiv} version of this paper \cite{AbR2}. Recently, the same gradient ascent algorithm as above appeared in \cite{WXGMD} with weighted sum-rate as the objective function instead of sum-rate as the objective (as considered here).

Note that evaluation of $\nabla_V f$ (in Step $8$ of Algorithm $1$) and the function $f$ (in Step $13$ of Algorithm $1$) have high complexity because they require averaging over an arbitrary number of Gaussian noise samples. This motivates us to pursue low complexity gradient ascent based algorithms in the following subsection.

\subsection{Low Complexity Gradient Ascent Algorithms based on high SNR approximation} \label{subsec4a}
In this subsection, we assume uniform distribution over the elements of the finite constellation. Since the high SNR approximation in (\ref{actual_approx}) does not involve averaging over noise samples, we propose to maximize the objective function defined by
\begin{align} \nonumber
& f_1\left(V_1, \cdots, V_K\right)=\sum_{i=1}^{3} log~M^{d_i} \\
\label{approx1_obj}
&-\frac{1}{M^{\sum_{k=1}^{K}d_k}}\hspace{-0.2cm}\sum_{k_1 = 0}^{M^{\sum_{k=1}^{K}d_k}-1}\left[ log  \left( \sum_{k_2 = 0}^{M^{\sum_{k=1}^{K}d_k}-1}e^{ - \left|\left|\sqrt P A^{k_1,k_2}_i\right|\right|^2 }\right) \right]\\
\nonumber
&+\frac{1}{ M^{\sum_{j \neq i}d_j}}\hspace{-0.2cm}\sum_{i_1 = 0}^{ M^{\sum_{j \neq i}d_j}-1}\left[ log  \left( \sum_{i_2 = 0}^{ M^{\sum_{j \neq i}d_j}-1}e^{ - \left|\left|\sqrt P B^{i_1,i_2}_i\right|\right|^2 }\right) \right].
\end{align}using Algorithm $1$ where, $f$ is replaced by $f_1$ and $\nabla_V f$ is replaced by  $\nabla_V f_1$.

We now interpret what maximizing $f_1$ and $f$ means in terms of the input constellations, at large values of $P$. The following observation would generalize the result of \cite{PRS} that the precoding matrix $V_i$ that maximizes $I[X_i;Y_i]$ when the interference channel matrices are zero converges to the matrix that maximizes the minimum distance between the desired constellation vectors, at large values of $P$.

\begin{theorem} \label{thm3}
At large values of $P$, the precoding matrices that maximize $f(V_1,\cdots,V_K)$ and $f_1(V_1,\cdots,V_K)$ converge to the matrices that maximize $\min_{i=1}^{K}d_{min}(i)$ where, $d_{min}(i)$ is given by

{\footnotesize \vspace{-0.2cm}
\begin{align} \label{eqn-define_dmin(i)}
 \min_{\substack{{p_{i1}\neq p_{i2},}\\{ \left(p_{k1}, p_{k2}\right)}}}\left|\left|H_{ii}V_i(x_i^{p_{i1}}-x_i^{p_{i2}})+\sum_{\substack{{k \neq i}\\{k=1}}}^{K}H_{kj}V_k(x_k^{p_{k1}}-x_k^{p_{k2}})\right|\right|
\end{align}}represents the minimum among the distances between two points belonging to different clouds at Rx-$i$ (without the power scaling $P$).
\end{theorem}
\begin{proof}
The high SNR approximation for $I[X_i;Y_i]$ in (\ref{actual_approx}) can be re-written as (\ref{eqn-dmin_1}) using (\ref{eqn-use_in_dmin_1}) and (\ref{eqn-use_in_dmin_2}). Using first order expansion for logarithm function, at large $P$, we have the approximation in (\ref{eqn-dmin_2}). At large $P$, the exponential terms in (\ref{eqn-dmin_2}) are negligible compared to the other terms. Hence, the precoders that maximize $\sum_{i=1}^{K}I[X_i;Y_i]$ at large $P$ must have $\left|{\cal C}_{2i}^{p_{11} ,\cdots,p_{K1}}\right|=0$, for all $\left(p_{11},\cdots , p_{K1}\right)$ and $i$, i.e., the precoders must be CCSC optimal. Using the definitions of $A^{k_1,k_2}_i$ and $B^{k_1,k_2}_i$ given in (\ref{eqn-define_A}) and (\ref{eqn-define_B}) respectively, and the definition of ${\cal D}^{p_{11},\cdots,p_{i-1,1},p_{i+1,1},\cdots, p_{K1}}_i$ in (\ref{eqn-define_calD}), we have (\ref{eqn-dmin_3}). Splitting the second term of (\ref{eqn-dmin_3}) into two terms, one with $p_{i,2}=p_{i,1}$ and another with $p_{i,2} \neq p_{
i,1}$, we have (\ref{eqn-dmin_4}). Now, the third and the last term of (\ref{eqn-dmin_4}) are equal and thus, they cancel each other. Hence, we have (\ref{eqn-dmin_5}). At high $P$, the exponential term corresponding to $d_{min}(i)$, defined in (\ref{eqn-define_dmin(i)}), dominates the value of (\ref{eqn-dmin_5}). Suppose the number of tuples $(p_{i1}, \cdots, p_{K1}, p_{i2}, \cdots, p_{K2})$ contributing to $d_{min}(i)$ in the exponent term of (\ref{eqn-dmin_5}) be $t_i$ and let the corresponding values of $\left|{\cal D}^{p_{11},\cdots,p_{i-1,1},p_{i+1,1},\cdots, p_{K1}}_i\right|$ be denoted by $D_{ik}$, for $k=1,2,\cdots,t_i$. Therefore, we have
\begin{align*}
 \sum_{i=1}^{3}I[X_i;Y_i] \approx \sum_{i=1}^{3}~log~M^{d_i} - \frac{e^{-Pd^2_{min}(i)} }{M^{\sum_{k=1}^{K}d_k}}\left(\sum_{k=1}^{t_i}\frac{log ~e}{1+D_{ik}}\right)
\end{align*}Again, at large $P$, the minimum among $d^2_{min}(i)$ dominates the value of the above expression. Hence, the precoding matrices that maximize $\sum_{i=1}^{3}I[X_i;Y_i]$ converge to the matrices that maximize $\min_{i=1}^{K}d_{min}(i)$, at large values of $P$.
\begin{figure*} \scriptsize
\begin{align}
 \nonumber
& I[X_i;Y_i] \approx log~M^{d_i} - \frac{1}{M^{\sum_{k=1}^{K}d_k}}\sum_{k_1 = 0}^{M^{\sum_{k=1}^{K}d_k}-1} log  \left(1+\frac{1}{1+|{\cal A}^{k_1}_i|}\sum_{\substack{{k_2 \neq k_1}\\{k_2 \notin {\cal A}^{k_1}_i}}} e^{ -P\left( \left|\left| A^{k_1,k_2}_i\right|\right|^2\right)} \right)\\
 \nonumber
&-\frac{1}{M^{\sum_{k=1}^{K}d_k}}\sum_{p_{11} = 0}^{M^{d_1}-1} \cdots \sum_{p_{i-1,1} = 0}^{M^{d_{i-1}}-1}\sum_{p_{i,1} = 0}^{M^{d_{i}}-1}\sum_{p_{i+1,1} = 0}^{M^{d_{i+1}}-1} \cdots \sum_{p_{K1} = 0}^{M^{d_K}-1}log  \left(1+|{\cal D}^{p_{11},\cdots,p_{i-1,1},p_{i+1,1},\cdots, p_{K1}}_i|+|{\cal C}^{p_{11},\cdots,p_{K1}}_{2i}| \right)\\
 \nonumber
&+\frac{1}{ M^{\sum_{j \neq i}d_j}}\sum_{p_{11} = 0}^{M^{d_1}-1} \cdots \sum_{p_{i-1,1} = 0}^{M^{d_{i-1}}-1}\sum_{p_{i+1,1} = 0}^{M^{d_{i+1}}-1}\cdots \sum_{p_{K1} = 0}^{M^{d_K}-1} log  \left(1+|{\cal D}^{p_{11},\cdots,p_{i-1,1},p_{i+1,1},\cdots, p_{K1}}_i| \right)\\
\label{eqn-dmin_1}
&+ \frac{1}{M^{\sum_{j \neq i}d_j}}\sum_{i_1 = 0}^{M^{\sum_{j \neq i}d_j}-1} log  \left(1+\frac{1}{1+|{\cal B}^{i_1}_i|}\sum_{\substack{{i_2 \neq i_1}\\{i_2 \notin {\cal B}^{i_1}_i}}} e^{ -P\left( \left|\left| B^{i_1,i_2}_i\right|\right|^2\right)} \right)\\ 
\nonumber
\approx & ~log~M^{d_i} - \frac{1}{M^{\sum_{i=1}^{K}d_i}}\sum_{k_1 = 0}^{M^{\sum_{i=1}^{K}d_i}-1} \frac{log ~e}{1+|{\cal A}^{k_1}_i|}\sum_{\substack{{k_2 \neq k_1}\\{k_2 \notin {\cal A}^{k_1}_i}}} e^{ -P\left( \left|\left| A^{k_1,k_2}_i\right|\right|^2\right)} \\
\nonumber
&-\frac{1}{M^{\sum_{k=1}^{K}d_k}}\sum_{p_{11} = 0}^{M^{d_1}-1} \cdots \sum_{p_{i-1,1} = 0}^{M^{d_{i-1}}-1}\sum_{p_{i,1} = 0}^{M^{d_{i}}-1}\sum_{p_{i+1,1} = 0}^{M^{d_{i+1}}-1} \cdots \sum_{p_{K1} = 0}^{M^{d_K}-1}log  \left(1+|{\cal D}^{p_{11},\cdots,p_{i-1,1},p_{i+1,1},\cdots, p_{K1}}_i|+|{\cal C}^{p_{11},\cdots,p_{K1}}_{2i}| \right)\\
\nonumber
&+\frac{1}{ M^{\sum_{j \neq i}d_j}}\sum_{p_{11} = 0}^{M^{d_1}-1} \cdots \sum_{p_{i-1,1} = 0}^{M^{d_{i-1}}-1}\sum_{p_{i+1,1} = 0}^{M^{d_{i+1}}-1}\cdots \sum_{p_{K1} = 0}^{M^{d_K}-1} log  \left(1+|{\cal D}^{p_{11},\cdots,p_{i-1,1},p_{i+1,1},\cdots, p_{K1}}_i| \right)\\
\label{eqn-dmin_2}
&+  \frac{1}{M^{\sum_{j \neq i}d_j}}\sum_{i_1 = 0}^{M^{\sum_{j \neq i}d_j}-1} \frac{log ~e}{1+|{\cal B}^{i_1}_i|}\sum_{\substack{{i_2 \neq i_1}\\{i_2 \notin {\cal B}^{i_1}_i}}} e^{ -P\left( \left|\left| B^{i_1,i_2}_i\right|\right|^2\right)}\\ 
\nonumber
=&~log~M^{d_i} -\frac{1}{M^{\sum_{k=1}^{K}d_k}}\sum_{p_{11} = 0}^{M^{d_1}-1} \cdots \sum_{p_{i-1,1} = 0}^{M^{d_{i-1}}-1}\sum_{p_{i,1} = 0}^{M^{d_{i}}-1}\sum_{p_{i+1,1} = 0}^{M^{d_{i+1}}-1} \cdots \sum_{p_{K1} = 0}^{M^{d_K}-1}\frac{log ~e}{1+|{\cal D}^{p_{11},\cdots,p_{i-1,1},p_{i+1,1},\cdots, p_{K1}}|}\\
\nonumber
&\hspace{3cm}\sum_{p_{12} = 0}^{M^{d_1}-1} \cdots \sum_{p_{i-1,2} = 0}^{M^{d_{i-1}}-1}\sum_{p_{i,2} = 0}^{M^{d_{i}}-1}\sum_{p_{i+1,2} = 0}^{M^{d_{i+1}}-1} \cdots \sum_{p_{K2} = 0}^{M^{d_K}-1} e^{-P\left|\left| H_{ii}V_i(x_i^{p_{i1}}-x_i^{p_{i2}})+\sum_{\substack{{k \neq i}\\{k=1}}}^{K}H_{kj}V_k(x_k^{p_{k1}}-x_k^{p_{k2}}) \right| \right|^2}\\
\nonumber
&+\frac{1}{M^{\sum_{j\neq i}d_j}}\sum_{p_{11} = 0}^{M^{d_1}-1} \cdots \sum_{p_{i-1,1} = 0}^{M^{d_{i-1}}-1}\sum_{p_{i+1,1} = 0}^{M^{d_{i+1}}-1} \cdots \sum_{p_{K1} = 0}^{M^{d_K}-1}\frac{log ~e}{1+|{\cal D}^{p_{11},\cdots,p_{i-1,1},p_{i+1,1},\cdots, p_{K1}}_i|}\\
\label{eqn-dmin_3}
&\hspace{3cm}\sum_{p_{12} = 0}^{M^{d_1}-1} \cdots \sum_{p_{i-1,2} = 0}^{M^{d_{i-1}}-1}\sum_{p_{i+1,2} = 0}^{M^{d_{i+1}}-1} \cdots \sum_{p_{K2} = 0}^{M^{d_K}-1} e^{-P\left|\left| \sum_{\substack{{k=1}\\{k \neq i}}}^{K}H_{kj}V_k(x_k^{p_{k1}}-x_k^{p_{k2}}) \right| \right|^2}\\
\nonumber
=&~log~M^{d_i} -\frac{1}{M^{\sum_{k=1}^{K}d_k}}\sum_{p_{11} = 0}^{M^{d_1}-1} \cdots \sum_{p_{i-1,1} = 0}^{M^{d_{i-1}}-1}\sum_{p_{i,1} = 0}^{M^{d_{i}}-1}\sum_{p_{i+1,1} = 0}^{M^{d_{i+1}}-1} \cdots \sum_{p_{K1} = 0}^{M^{d_K}-1}\frac{log ~e}{1+|{\cal D}^{p_{11},\cdots,p_{i-1,1},p_{i+1,1},\cdots, p_{K1}}_i|}\\
\nonumber
&\hspace{3cm}\sum_{p_{12} = 0}^{M^{d_1}-1} \cdots \sum_{p_{i-1,2} = 0}^{M^{d_{i-1}}-1}\sum_{\substack{{p_{i,2} = 0}\\{p_{i,2}\neq p_{i,1}}}}^{M^{d_{i}}-1}\sum_{p_{i+1,2} = 0}^{M^{d_{i+1}}-1} \cdots \sum_{p_{K2} = 0}^{M^{d_K}-1} e^{-P\left|\left| H_{ii}V_i(x_i^{p_{i1}}-x_i^{p_{i2}})+\sum_{\substack{{k \neq i}\\{k=1}}}^{K}H_{kj}V_k(x_k^{p_{k1}}-x_k^{p_{k2}}) \right| \right|^2}\\
\nonumber
&-\frac{1}{M^{\sum_{k=1}^{K}d_k}}\sum_{p_{11} = 0}^{M^{d_1}-1} \cdots \sum_{p_{i-1,1} = 0}^{M^{d_{i-1}}-1}\sum_{p_{i,1} = 0}^{M^{d_{i}}-1}\sum_{p_{i+1,1} = 0}^{M^{d_{i+1}}-1} \cdots \sum_{p_{K1} = 0}^{M^{d_K}-1}\frac{log ~e}{1+|{\cal D}^{p_{11},\cdots,p_{i-1,1},p_{i+1,1},\cdots, p_{K1}}_i|}\\
\nonumber
&\hspace{3cm}\sum_{p_{12} = 0}^{M^{d_1}-1} \cdots \sum_{p_{i-1,2} = 0}^{M^{d_{i-1}}-1}\sum_{p_{i,2} = p_{i,1}}\sum_{p_{i+1,2} = 0}^{M^{d_{i+1}}-1} \cdots \sum_{p_{K2} = 0}^{M^{d_K}-1} e^{-P\left|\left| H_{ii}V_i(x_i^{p_{i1}}-x_i^{p_{i2}})+\sum_{\substack{{k \neq i}\\{k=1}}}^{K}H_{kj}V_k(x_k^{p_{k1}}-x_k^{p_{k2}}) \right| \right|^2}\\
\nonumber
&+\frac{1}{M^{\sum_{j\neq i}d_j}}\sum_{p_{11} = 0}^{M^{d_1}-1} \cdots \sum_{p_{i-1,1} = 0}^{M^{d_{i-1}}-1}\sum_{p_{i+1,1} = 0}^{M^{d_{i+1}}-1} \cdots \sum_{p_{K1} = 0}^{M^{d_K}-1}\frac{log ~e}{1+|{\cal D}^{p_{11},\cdots,p_{i-1,1},p_{i+1,1},\cdots, p_{K1}}|}\\
\label{eqn-dmin_4}
&\hspace{3cm}\sum_{p_{12} = 0}^{M^{d_1}-1} \cdots \sum_{p_{i-1,2} = 0}^{M^{d_{i-1}}-1}\sum_{p_{i+1,2} = 0}^{M^{d_{i+1}}-1} \cdots \sum_{p_{K2} = 0}^{M^{d_K}-1} e^{-P\left|\left| \sum_{\substack{{k=1}\\{k \neq i}}}^{K}H_{kj}V_k(x_k^{p_{k1}}-x_k^{p_{k2}}) \right| \right|^2}\\
\nonumber
=&~log~M^{d_i} -\frac{1}{M^{\sum_{k=1}^{K}d_k}}\sum_{p_{11} = 0}^{M^{d_1}-1} \cdots \sum_{p_{i-1,1} = 0}^{M^{d_{i-1}}-1}\sum_{p_{i,1} = 0}^{M^{d_{i}}-1}\sum_{p_{i+1,1} = 0}^{M^{d_{i+1}}-1} \cdots \sum_{p_{K1} = 0}^{M^{d_K}-1}\frac{log ~e}{1+|{\cal D}^{p_{11},\cdots,p_{i-1,1},p_{i+1,1},\cdots, p_{K1}}_i|}\\
\label{eqn-dmin_5}
&\hspace{3cm}\sum_{p_{12} = 0}^{M^{d_1}-1} \cdots \sum_{p_{i-1,2} = 0}^{M^{d_{i-1}}-1}\sum_{\substack{{p_{i,2} = 0}\\{p_{i,2}\neq p_{i,1}}}}^{M^{d_{i}}-1}\sum_{p_{i+1,2} = 0}^{M^{d_{i+1}}-1} \cdots \sum_{p_{K2} = 0}^{M^{d_K}-1} e^{-P\left|\left| H_{ii}V_i(x_i^{p_{i1}}-x_i^{p_{i2}})+\sum_{\substack{{k \neq i}\\{k=1}}}^{K}H_{kj}V_k(x_k^{p_{k1}}-x_k^{p_{k2}}) \right| \right|^2}
\end{align}
\hrule
\end{figure*}
\end{proof}

We observe that maximizing $\min_{i=1}^{K}d_{min}(i)$ also minimizes the probability of error with ML decoding across all the receivers, at large values of $P$.

Now, note that at large values of $P$, we can further approximate (\ref{actual_approx}) by 
\begin{align} \nonumber
&I\left[X_i;Y_i\right] \approx log~M^{d_i} \\
\label{approx1_obj}
&-\frac{1}{M^{\sum_{k=1}^{K}d_k}}\hspace{-0.2cm}\sum_{k_1 = 0}^{M^{\sum_{k=1}^{K}d_k}-1}\left[ log  \left( \sum_{k_2 = 0}^{M^{\sum_{k=1}^{K}d_k}-1}e^{ - \frac{\left|\left|\sqrt P A^{k_1,k_2}_i\right|\right|^2}{r} }\right) \right]\\
\nonumber
&+\frac{1}{ M^{\sum_{j \neq i}d_j}}\hspace{-0.2cm}\sum_{i_1 = 0}^{ M^{\sum_{j \neq i}d_j}-1}\left[ log  \left( \sum_{i_2 = 0}^{ M^{\sum_{j \neq i}d_j}-1}e^{ - \frac{\left|\left|\sqrt P B^{i_1,i_2}_i\right|\right|^2}{r} }\right) \right].
\end{align}for some positive real number $r$. This is because the above expression and the expression in (\ref{actual_approx}) tend to the same value as $P$ tends to infinity. 

Thus, another low complexity algorithm shall involve maximization of the objective function given by 
\begin{align} \nonumber
& f_2\left(V_1, \cdots, V_K\right)= \sum_{i=1}^{3}log~M^{d_i} \\
\label{approx1_obj}
&-\frac{1}{M^{\sum_{k=1}^{K}d_k}}\hspace{-0.2cm}\sum_{k_1 = 0}^{M^{\sum_{k=1}^{K}d_k}-1}\left[ log  \left( \sum_{k_2 = 0}^{M^{\sum_{k=1}^{K}d_k}-1}e^{ - \frac{\left|\left|\sqrt P A^{k_1,k_2}_i\right|\right|^2}{r} }\right) \right]\\
\nonumber
&+\frac{1}{ M^{\sum_{j \neq i}d_j}}\hspace{-0.2cm}\sum_{i_1 = 0}^{ M^{\sum_{j \neq i}d_j}-1}\left[ log  \left( \sum_{i_2 = 0}^{ M^{\sum_{j \neq i}d_j}-1}e^{ - \frac{\left|\left|\sqrt P B^{i_1,i_2}_i\right|\right|^2}{r} }\right) \right].
\end{align}using Algorithm $1$ where, $f$ is replaced by $f_2$ and $\nabla_V f$ is replaced by  $\nabla_V f_2$. Following the chain rule for matrix differentials \cite{AH}, the gradient $\nabla_{V_k} f_2$, for $k=1,2,\cdots,K$, is given by (\ref{eqn-approx_grad}) (at the top of the page after the next page) where, $x_{k,k_1,k_2}$ is the sub-vector of the vector $A^{k_1,k_2}$ which corresponds to $x^{p_1}_k-x^{p_2}_k$, for some $p_1,p_2=0,1,\cdots,M^{d_k}-1$, and $x_{k,i_1,i_2}$ is the sub-vector of the vector $B^{i_1,i_2}$ which corresponds to $x^{p_1}_k-x^{p_2}_k$, for some $p_1,p_2=0,1,\cdots,M^{d_k}-1$.
\begin{figure*}
\begin{align} \nonumber
\nabla_{V_k} f_2~=~& \frac{1}{r}\sum_{i=1}^{K} \sum_{k_1=0}^{M^{\sum_{i'=1}^{K} d_{i'}}-1} \left(\frac{\sum_{k_2=0}^{M^{\sum_{i'=1}^{K} d_{i'}}-1} ~~e^{-\frac{\left|\left|A^{k_1,k_2}_i\right|\right|^2}{r}}H^H_{ki}A^{k_1,k_2}_ix^H_{k,k_1,k_2}}{\sum_{k_2=0}^{M^{\sum_{i'=1}^{K} d_{i'}}-1} ~~e^{-\frac{\left|\left|A^{k_1,k_2}_i\right|\right|^2}{r}}}\right) \\&-	
\label{eqn-approx_grad}
\frac{1}{r}\sum_{\substack{{j=1}\\{j\neq k}}}^{K} \sum_{i_1=0}^{M^{\sum_{j'=1}^{K} d_{j'}}-1} \left(\frac{\sum_{i_2=0}^{M^{\sum_{j'=1}^{K} d_{j'}}-1} ~~e^{-\frac{\left|\left|B^{i_1,i_2}_j\right|\right|^2}{r}}H^H_{kj}B^{i_1,i_2}_jx^H_{k,i_1,i_2}}{\sum_{i_2=0}^{M^{\sum_{j'=1}^{K} d_{j'}}-1} ~~e^{-\frac{\left|\left|B^{i_1,i_2}_j\right|\right|^2}{r}}}\right) 
\end{align}
\hrule
\end{figure*}

In the following subsection, we present some simulation results using the proposed algorithms.

\subsection{Simulation Results} \label{subsec4b}
Several algorithms to obtain precoders that aim to achieve IA are known. We consider two representatives of such algorithms from \cite{GCJ} and \cite{SPLL} for comparison with the proposed algorithms. The works in \cite{GCJ} and \cite{SPLL} demonstrate the performance of their algorithms in terms of sum-rate with Gaussian alphabet inputs. In this section, we present examples of performance of these algorithms with the practical case of finite constellation inputs.

Consider a MIMO GIC with $K=3$, $n_{t_i}=n_{r_i}=2$, $d_i=1$ for all $i$. We shall consider the max-SINR (signal to interference plus noise ratio) algorithm from \cite{GCJ} and the maximum sum chordal distance algorithm from Section IV A of \cite{SPLL} for comparison with the proposed algorithms. The closed form IA precoder solution of \cite{CaJ1} for the considered MIMO GIC is given by 
\begin{align*}
 V_1=e_1, ~V_2=H^{-1}_{23} H_{13}V_1, ~V_3=H^{-1}_{32}H_{12}V_1
\end{align*}where, $e_1$ is an eigen vector of the matrix $H^{-1}_{13}H_{23}H^{-1}_{21}H_{31}H^{-1}_{32}H_{12}$. The maximum sum chordal distance algorithm selects the eigen vector $e_1$ which maximizes the sum chordal distance \cite{SPLL}. We however select the eigen vector $e_1$ which maximizes the sum-rate $\sum_{i=1}^{3}I[X_i;Y_i]$ with finite constellation inputs. Clearly, the maximum sum chordal distance algorithm cannot perform better in terms of sum-rate when $e_1$ is chosen to maximize the sum-rate with finite constellation inputs. We call this as the max-sum-rate CaJ IA solution. The max-SINR algorithm aims to maximize the signal to interference plus noise ratio at each of the receivers so that the sum-rate with $I[X_i;Y_i]$ with Gaussian input alphabets is maximized. Unlike in \cite{GCJ} or \cite{SPLL}, we do not use any receive filter matrices as using them can only reduce the rate $I[X_i;Y_i]$ because of data-processing inequality \cite{CoT}. The max-SINR algorithm computes receive filter 
matrices at every iteration\footnote{For the sake of brevity, we do not present the details of the algorithm and the reader can have the details from \cite{GCJ}.}. We discard the receive filter matrices once the max-SINR algorithm converges and then compute $\sum_{i=1}^{3}I[X_i;Y_i]$ with finite constellation inputs. Similarly, we compute $\sum_{i=1}^{3}I[X_i;Y_i]$ with finite constellation inputs for the max-sum-rate CaJ IA solution without the use of any zero forcing filters at the receivers. 

\begin{enumerate}
\item {\em Performance of the proposed algorithm averaged over channel realizations:} The ergodic sum-rates, i.e., the sum-rate averaged with the entries of the channel matrices being taken from ${\cal CN}(0,1)$, using BPSK and QPSK input constellations are simulated. The ergodic sum-rates obtained using the max-SINR algorithm (with random initialization), max-sum-rate CaJ IA solution, and the proposed gradient ascent algorithms with $f$, $f_1$, and $f_2$ as the objective functions using the max-SINR precoders as the initialization are shown for the BPSK constellation and the QPSK constellation in Fig. \ref{fig-sum_rate_comparison}$(a)$ and in Fig. \ref{fig-sum_rate_comparison}$(b)$ respectively. The parameter $r$ chosen for the objective function $f_2$ is equal to $2$. The chosen parameters in the gradient ascent algorithm are given by $\text{\em max\_iterations}=10, \beta=0.2, \alpha=0.005, \epsilon=0.001$. As seen from Fig. \ref{fig-sum_rate_comparison}$(a)$ and in Fig. \ref{fig-sum_rate_comparison}$(b)$, 
the precoders obtained by optimization of $f_1$ gives negligible improvement in the ergodic sum-rate over that obtained from the max-SINR algorithm. However the precoders obtained by optimization of $f$ and $f_2$ give considerable improvement in the ergodic sum-rate over that obtained from the max-SINR algorithm. Moreover, there is negligible difference in the sum-rate obtained by optimizing $f$ and $f_2$. Hence, gradient ascent algorithm with $f_2$ as the objective function (with $r=2$) is a worthy low-complexity alternative to the gradient ascent algorithm with $f$ as the objective function. Also, observe that the max-sum-rate CaJ IA solution performs badly compared to the other algorithms. For clarity on the gains in the ergodic sum-rate obtained by optimizing $f$ and $f_2$ over that obtained from the max-SINR algorithm, the ergodic sum-rate values are given in Table \ref{tab1} and Table \ref{tab2} for BPSK and QPSK constellations respectively.
\begin{figure*}[htbp]
\centering
\subfigure[BPSK constellation ] {\label{subfig:1}\includegraphics[totalheight=3.0in,width=5.6in]{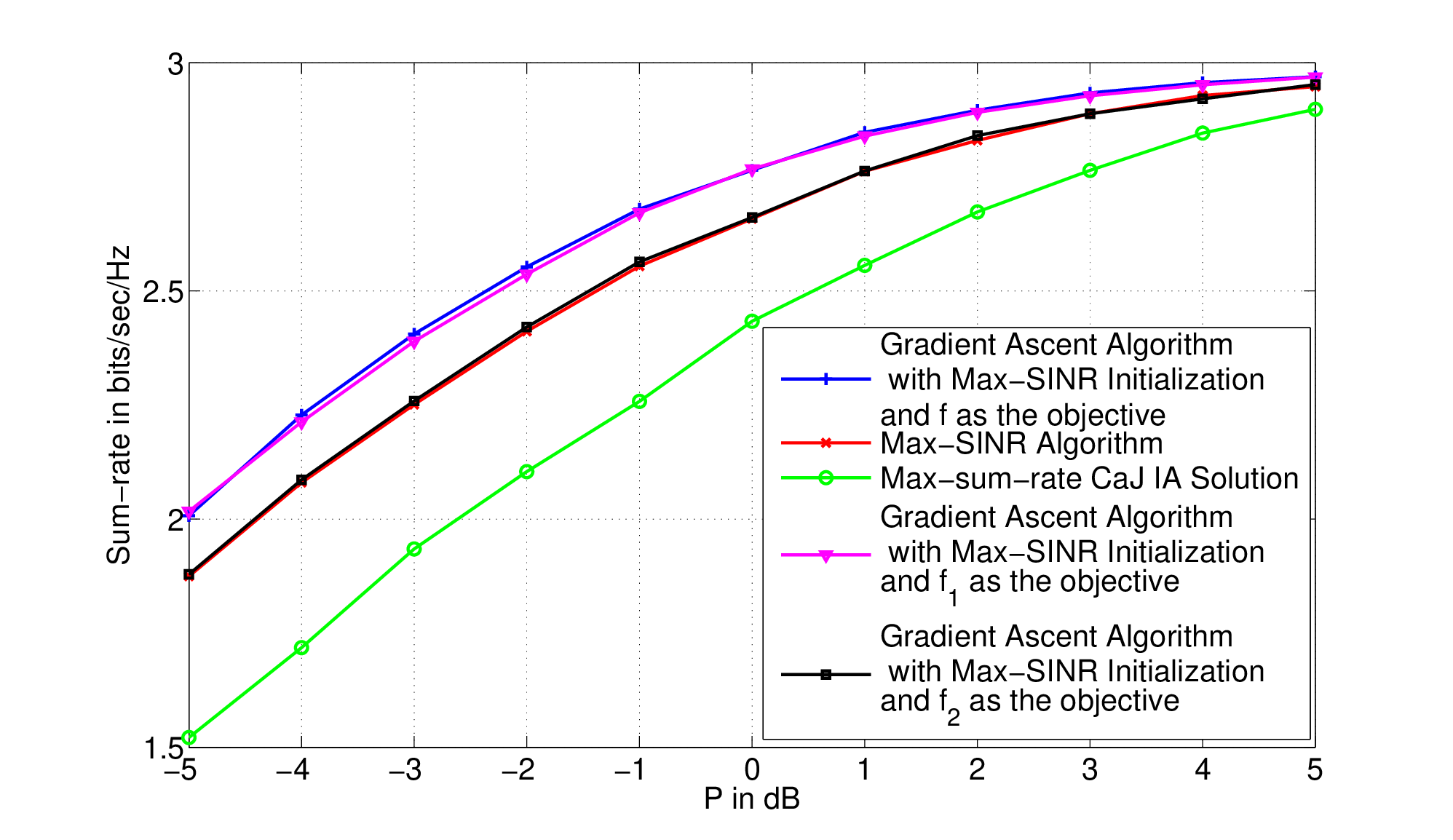}}
\subfigure[QPSK constellation ] {\label{subfig:2}\includegraphics[totalheight=3.0in,width=5.6in]{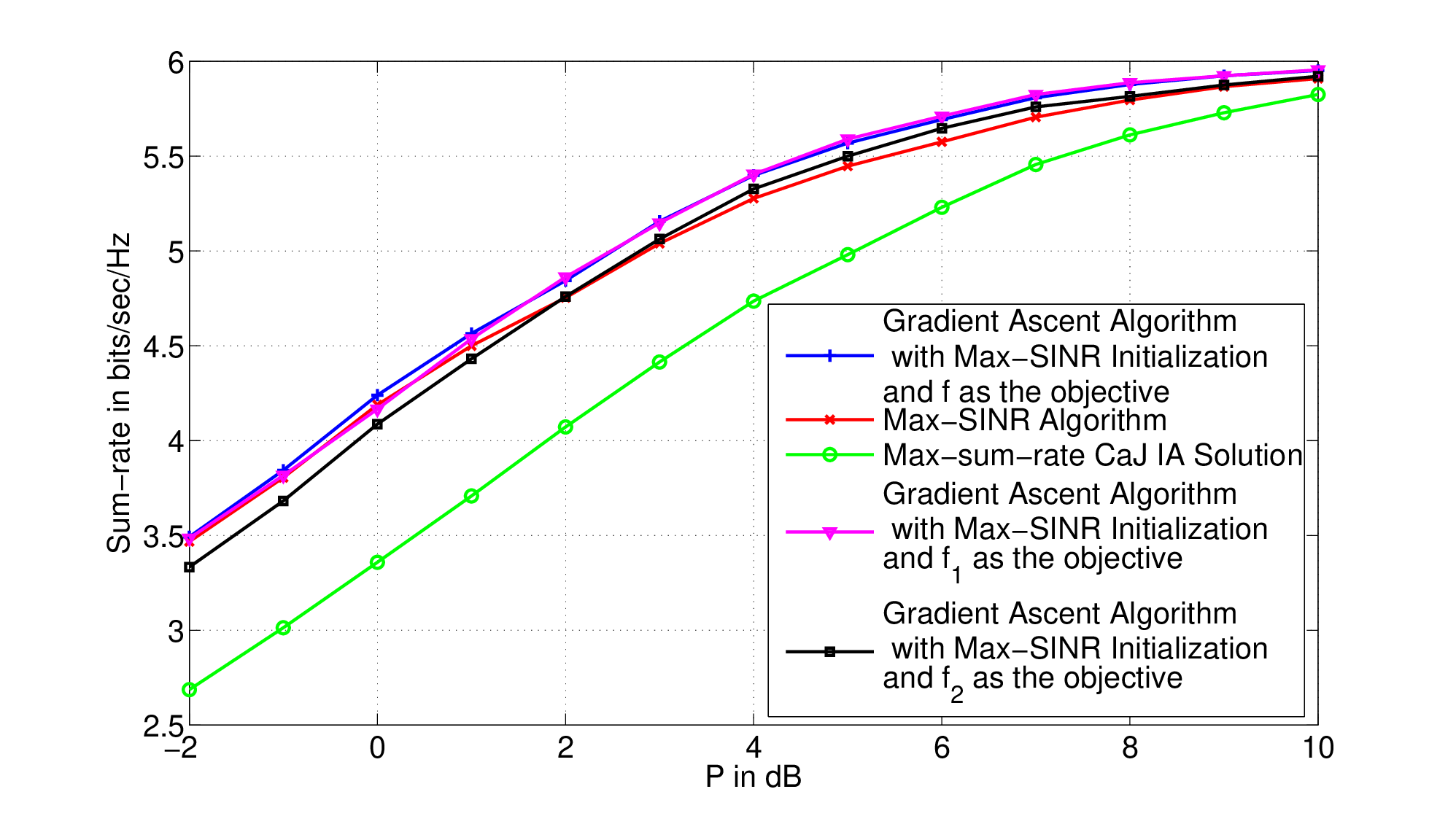}}
\caption{Sum-rate (in bits/sec/Hz) vs P (in dB).}
 \label{fig-sum_rate_comparison}
 \hrule
\end{figure*}
\begin{table*} 
\centering
\normalsize
\caption{Ergodic sum-rate values (in bits/sec/Hz) for BPSK constellation inputs}
\begin{tabular}{|c|c|c|c|c|c|c|c|c|c|c|c|} 
\hline
\small P (in dB) & $-5$ & $-4$ & $-3$ & $-2$ & $-1$ & $0$ & $1$ & $2$ & $3$ & $4$ & $5$\\
\hline
\small Max-SINR algorithm & $1.874$ & $2.079$ & $2.251$ & $2.411$ & $2.554$ &  $2.658$ & $2.762$ & $2.829$ & $2.888$ & $2.928$ & $2.947$\\
\hline
\small Gradient ascent algorithm  &  &  &  &  &  &  &  &  &  &  &  \\
\small with $f$ as the objective & $2.007$ & $2.228$ & $2.405$ & $2.552$ & $2.679$ & $2.764$ & $2.847$ & $2.895$ & $2.933$ & $2.956$ &  $2.969$\\
\hline
\small Gradient ascent algorithm &  &  &  &  &  &  &  &  &  &  & \\
\small with $f_2$ as the objective & $2.027$ & $2.206$ & $2.394$ & $2.564$ & $2.689$ & $2.767$ & $2.846$ & $2.895$ & $2.931$ & $2.952$ & $2.970$\\
\hline
\small Max-sum-rate  &  &  &  &  &  &  &  &  &  &  &  \\
 \small CaJ IA Solution & $1.521$ & $1.718$ & $1.934$ & $2.104$ & $2.258$ & $2.433$ & $2.555$ & $2.673$ & $2.764$ & $2.845$ & $2.898$\\
\hline
\end{tabular}
~~\\
~~\\
~\\
~~\\
~\\
~~\\
~\\
~~\\
\label{tab1}
\end{table*}
\begin{table*} 
\centering
\normalsize
\caption{Ergodic sum-rate values (in bits/sec/Hz) for QPSK constellation inputs}
\begin{tabular}{|c|c|c|c|c|c|c|c|c|c|c|c|c|c|} 
\hline
\footnotesize P (in dB) & \footnotesize$-2$ & \footnotesize$-1$ & \footnotesize$0$ & \footnotesize$1$ & \footnotesize$2$ &\footnotesize $3$ & \footnotesize$4$ & \footnotesize$5$ & \footnotesize$6$ & \footnotesize$7$ & \footnotesize$8$ & \footnotesize$9$ & \footnotesize$10$\\
\hline
\footnotesize Max-SINR &  &  &  &  &  &  &  &  &  &  & & & \\
\footnotesize algorithm & \footnotesize$3.463$ & \footnotesize$3.784$ & \footnotesize$4.117$ & \footnotesize$4.471$ & \footnotesize$4.771$ & \footnotesize$5.032$ & \footnotesize$5.273$ & \footnotesize$5.463$ & \footnotesize$5.589$ & \footnotesize$5.732$ & \footnotesize$5.810$ & \footnotesize$5.866$ & \footnotesize$5.919$\\
\hline
\footnotesize Gradient ascent &  &  &  &  &  &  &  &  &  &  & & & \\
\footnotesize algorithm  &  &  &  &  &  &  &  &  &  &  & & & \\
\footnotesize with $f$ as the objective & \footnotesize$3.492$ & \footnotesize$3.843$ & \footnotesize$4.239$ & \footnotesize$4.564$ & \footnotesize$4.843$ & \footnotesize$5.155$ & \footnotesize$5.398$ & \footnotesize$5.568$ & \footnotesize$5.692$ & \footnotesize$5.808$ & \footnotesize$5.877$ & \footnotesize$5.923$ & \footnotesize$5.951$\\
\hline
\footnotesize Gradient ascent &  &  &  &  &  &  &  &  &  &  & & & \\
\footnotesize algorithm &  &  &  &  &  &  &  &  &  &  & & &\\
\footnotesize with $f_2$ as the objective &  \footnotesize$3.485$ & \footnotesize$3.816$ & \footnotesize$4.166$ & \footnotesize$4.536$ & \footnotesize$4.862$ & \footnotesize$5.147$ & \footnotesize$5.404$ & \footnotesize$5.589$ & \footnotesize$5.710$ & \footnotesize$5.824$ & \footnotesize$5.886$ & \footnotesize$5.923$ & \footnotesize$5.954$\\
\hline
\footnotesize Max-sum-rate  &  &  &  &  &  &  &  &  &  &  &  & &\\
 \footnotesize CaJ IA Solution &  \footnotesize$2.686$ & \footnotesize$3.012$ & \footnotesize$3.358$ & \footnotesize$3.708$ & \footnotesize$4.071$ & \footnotesize$4.414$ & \footnotesize$4.735$ & \footnotesize$4.980$ & \footnotesize$5.229$ & \footnotesize$5.455$ & \footnotesize$5.611$ & \footnotesize$5.728$ & \footnotesize$5.824$\\
\hline
\end{tabular}
~~\\
~~\\
~\\
\label{tab2}
\end{table*}
As observed from Table \ref{tab1}, the ergodic sum-rate gain using the proposed gradient ascent algorithm with $f$ and $f_2$ as objectives over the max-SINR algorithm for BPSK input constellations is more than $0.1$ bits/sec/Hz upto $P=0 ~dB$. As expected the gain decreases as $P$ becomes higher as the sum-rate obtained using all the algorithms saturate to $3$ bits/sec/Hz for every channel realizations. Similarly from Table \ref{tab2}, from $P=3~dB$ to $P=6~dB$, the ergodic sum-rate gain using the proposed gradient ascent algorithm with $f$ and $f_2$ as objective functions over the max-SINR algorithm for QPSK  input constellations is more than $0.1$ bits/sec/Hz. In this case also, the gain decreases as $P$ becomes higher. 

\item {\em Convergence and performance of the proposed algorithm for a fixed channel:} 
The parameters used in Algorithm $1$ are $\text{\em max\_iterations}=15, \beta=0.2, \alpha=0.005, \epsilon=0.001$. The chosen channel matrix is given in (\ref{ch_value}). The convergence behaviour of Algorithm $1$ using $f$ and $f_2$ (with $r=2$) as objective functions with BPSK inputs is shown in Fig. \ref{fig-conv} for $P=-5~dB$, $P=-2~dB$, and $P=0~dB$, with precoders obtained from the max-SINR algorithm as initialization.
The initial precoders for the max-SINR algorithm are chosen randomly. The proposed algorithm with $f_2$ as the objective function terminates well before the $\text{\em max\_iterations}$ number for all $P$ because the condition in Line $5$ of Algorithm $1$ is satisfied. As seen from Fig. \ref{fig-conv}, in all the cases the sum-rates obtained at the termination of the proposed algorithm with $f_2$ as the objective function are almost the same as that obtained with $f$ as the objective function. Furthermore, for all $P$, the sum-rate gains obtained over the max-SINR algorithm is more than $0.1$ bits/sec/Hz. 
\begin{figure*}
\scriptsize
\begin{align}
\label{ch_value}
&H= \begin{bmatrix}
	     0.3109 - 0.3888i &  0.3610 - 0.1670i & -0.2818 - 0.4540i &  0.4015 + 0.1563i & -0.8145 + 0.3811i &  0.3374 - 0.9180i\\
             0.3560 - 0.5511i &  1.1616 - 0.6200i &  0.6564 - 1.0746i &  0.0047 - 0.2788i & -0.2604 - 0.2678i & -0.4395 + 0.3621i\\
            -0.4314 + 0.3680i & -0.4350 - 0.5917i & -0.5202 - 1.2342i & -0.4241 + 0.5924i &  1.0494 - 0.6468i & -1.3259 + 0.0483i\\
            -0.2503 - 0.7360i & -0.4445 - 0.4758i & -0.6053 - 2.2125i &  0.8310 + 0.2683i &  0.2765 - 1.2192i &  0.2176 + 0.4875i\\
            -0.3055 + 0.4185i & -0.1248 - 0.6503i &  1.2821 + 0.3859i &  0.7999 + 1.0462i &  0.9247 - 0.9696i &  0.0276 - 0.1582i\\
             0.1163 - 0.2062i &  0.8211 - 0.4995i &  0.6084 - 0.6892i & -1.2459 + 0.1684i & -0.4081 + 1.2450i & -0.5386 - 0.1936i
    \end{bmatrix}.
\end{align}
\hrule
\vspace{-0.5cm}
\end{figure*}
\begin{figure*}[htbp]
\centering
\subfigure[$P=-5~dB$] {\label{subfig:1}\includegraphics[totalheight=2.5in,width=4.6in]{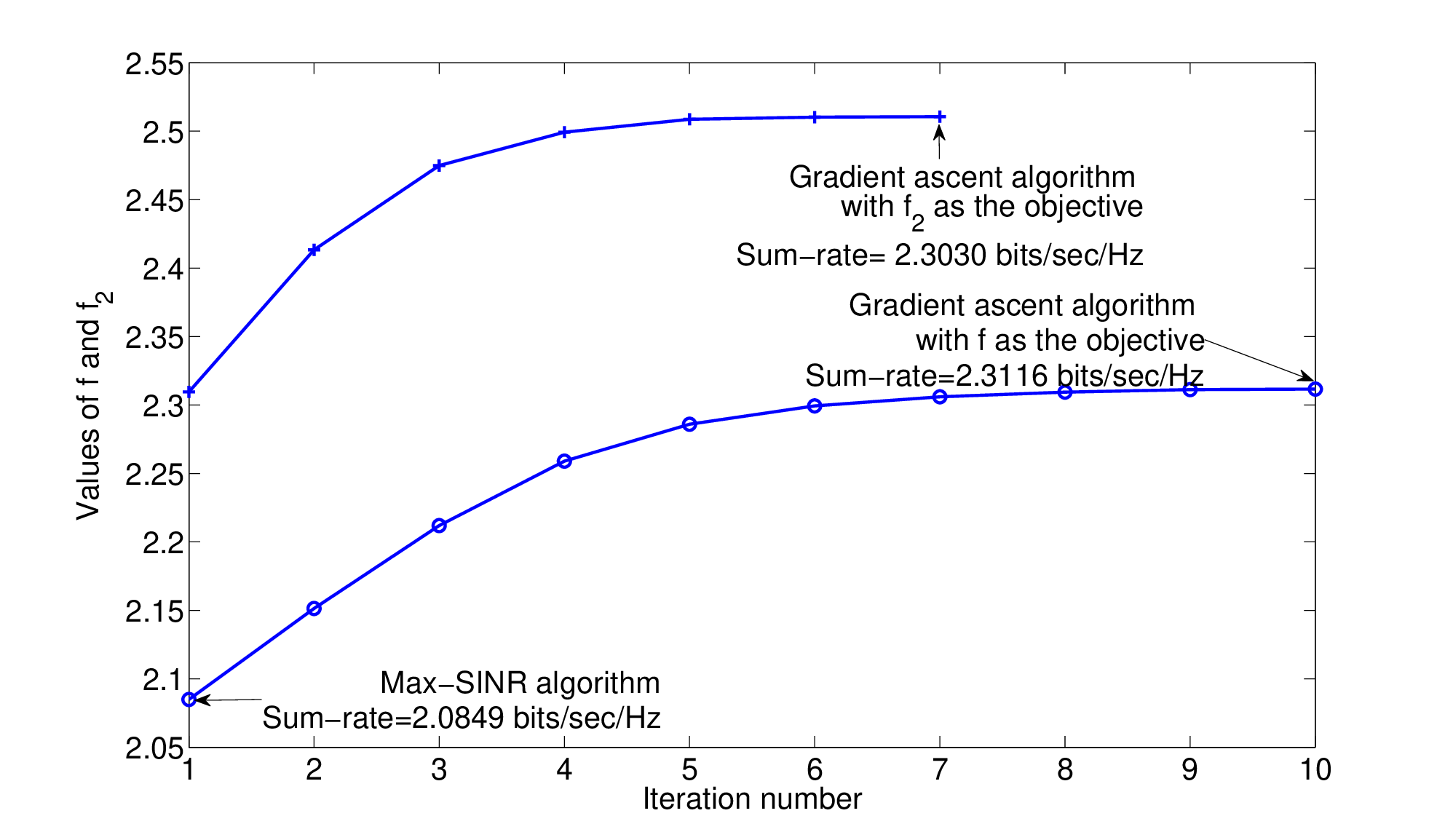}}\\
\subfigure[$P=-2~dB$] {\label{subfig:2}\includegraphics[totalheight=2.5in,width=4.6in]{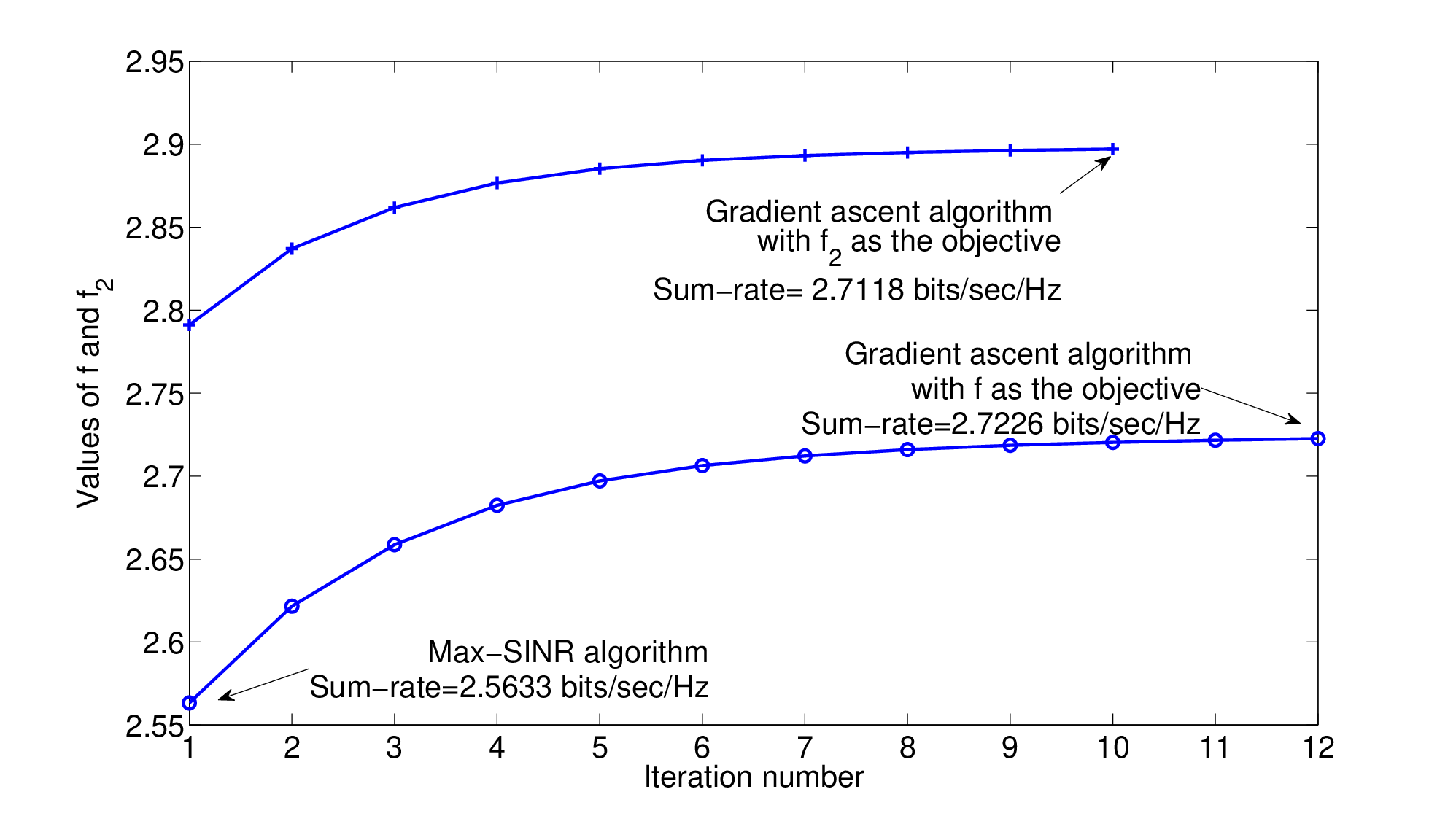}}\\
\subfigure[$P=0~dB$]  {\label{subfig:3}\includegraphics[totalheight=2.5in,width=4.6in]{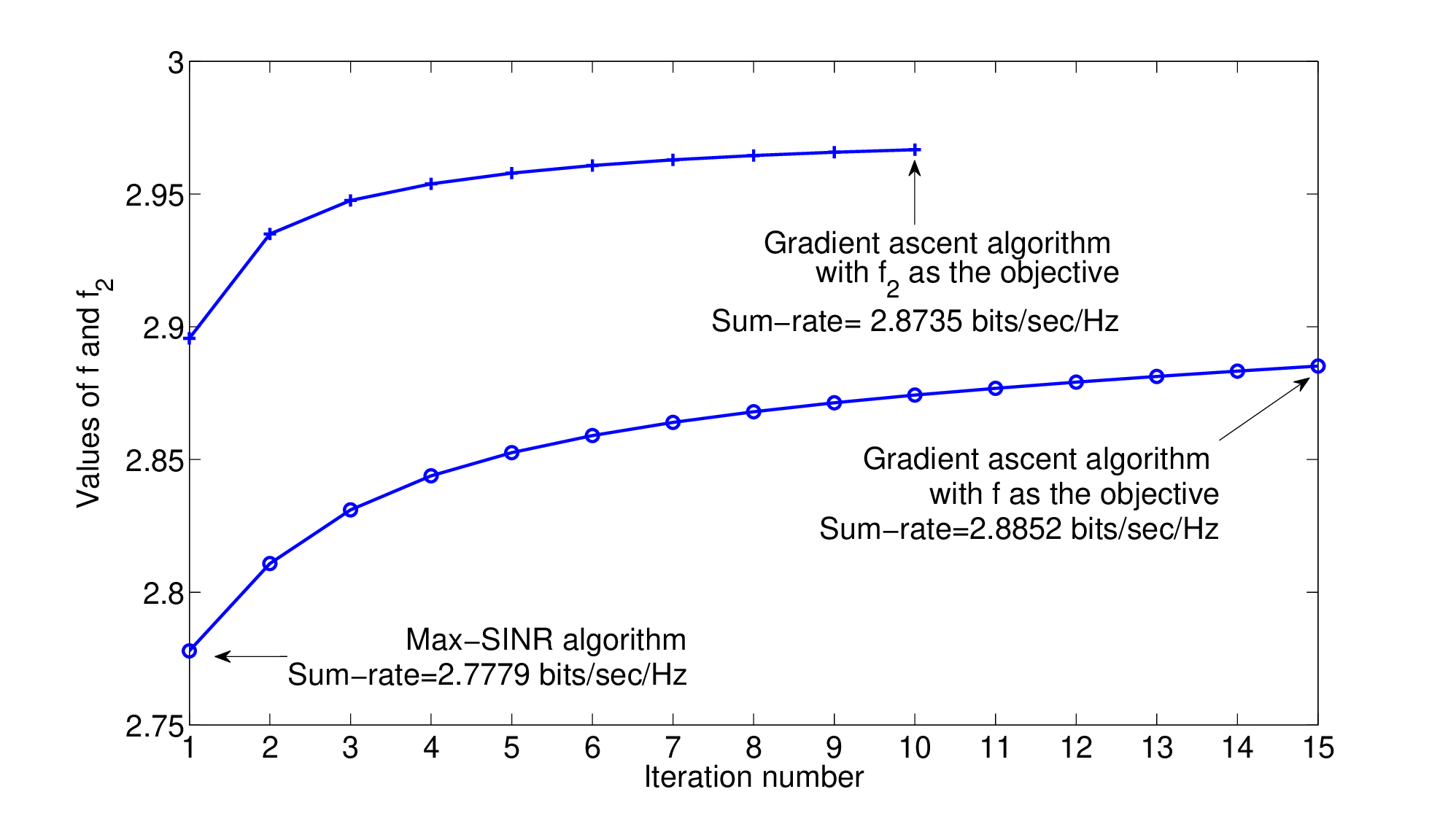}}
\caption{Increase in the objective functions $f$ and $f_2$ (with $r=2$) with every iteration for the constant channel case with BPSK inputs.}
\label{fig-conv}
\vspace{-0.5cm}
\end{figure*}
\end{enumerate}


\section{Conclusion} \label{sec5}
The paper discussed linear precoding for $K$-user MIMO GIC with finite constellation inputs. We showed that, for constant MIMO GIC with finite constellation inputs, CCSC for every transmitter can be achieved just by using a naive scheme of treating the interference as noise at every receiver, at high SNR. This result is in contrast with the Gaussian alphabet case where, at high SNR, the scheme that treats interference as noise saturates to a value determined by the channel gains for the SISO case. A set of necessary and sufficient conditions for CCSC optimal precoders were derived. It was observed that IA precoders fall under the class of CCSC optimal precoders. However, CCSC optimal precoders are easy to obtain for any given value of channel gains unlike obtaining IA precoders. Note that IA precoders have feasibility constraints which restrict the value of $d_i$, for all $i$. An important contribution of this paper is pointing out the CCSC optimality of IA precoders when the values of $d_i$ satisfy the 
feasibility constraints in \cite{GSB}. 

Finally, gradient ascent based algorithms with $f$ and $f_2$ as the objective functions were proposed. It was shown through simulations that optimizing the high SNR approximation for $f$, i.e., $f_2$ with $r=2$ performed as good as optimizing $f$ in terms of sum-rate. Thus, optimizing $f_2$ using the proposed gradient ascent based algorithm is a worthy low complexity algorithm. It was also observed that, at high SNR, optimizing $f$ or $f_2$ is equivalent to maximizing the minimum Euclidean distance for ML decoding across all the receivers.

\end{document}